\RequirePackage{xcolor}
\documentclass[journal,twoside,web]{ieeecolor}
\usepackage{generic}
\usepackage{cite}
\usepackage{amsmath,amssymb,amsfonts}
\usepackage{algorithm,algorithmic,subcaption}
\usepackage{graphicx}
\usepackage[OT2,T1]{fontenc}
\usepackage{MnSymbol}
\usepackage{lipsum}
\usepackage{pict2e,tikz,tkz-tab}
\usepackage{todonotes}
\let\labelindent\relax
\usepackage{enumitem}
\usepackage{cleveref}

\usepackage{amsthm}
\newcommand{\R}{\mathbb{R}}
\newcommand{\N}{\ensuremath{\mathbb{N}}}                 

\newcommand{\1}{{\bf  I}}
\newcommand{\0}{{\bf 0}}

\renewcommand{\ge}{\geqslant}
\renewcommand{\geq}{\geqslant}
\renewcommand{\le}{\leqslant}
\renewcommand{\leq}{\leqslant}
\renewcommand{\setminus}{\backslash}

\newcommand{\defeq}{\mathrel{\mathop:}=}

\DeclareMathOperator{\sign}{sign}
\DeclareMathOperator{\cl}{cl}
\DeclareMathOperator{\diag}{diag}
\DeclareMathOperator{\card}{card}
\DeclareMathOperator{\inter}{int}
\makeatletter
\newcommand{\subalign}[1]{%
	\vcenter{%
		\Let@ \restore@math@cr \default@tag
		\baselineskip\fontdimen10 \scriptfont\tw@
		\advance\baselineskip\fontdimen12 \scriptfont\tw@
		\lineskip\thr@@\fontdimen8 \scriptfont\thr@@
		\lineskiplimit\lineskip
		\ialign{\hfil$\m@th\scriptstyle##$&$\m@th\scriptstyle{}##$\hfil\crcr
			#1\crcr
		}%
	}%
}
\makeatother

\DeclareMathOperator*{\argmin}{arg\,min}

\newtheorem{theorem}{Theorem}
\newtheorem{assumption}{Assumption}
\newtheorem{definition}{Definition}

\newtheorem{proposition}{Proposition}

\newtheorem{lemma}{Lemma}
\newtheorem{corollary}{Corollary}
\newtheorem{remark}{Remark}
\DeclareSymbolFont{cyrletters}{OT2}{wncyr}{m}{n}
\DeclareMathSymbol{\Sha}{\mathalpha}{cyrletters}{"58}
\usepackage{textcomp}
\def\BibTeX{{\rm B\kern-.05em{\sc i\kern-.025em b}\kern-.08em
		T\kern-.1667em\lower.7ex\hbox{E}\kern-.125emX}}
\crefname{theorem}{Theorem}{Theorems}
\crefname{definition}{Definition}{Definitions}
\crefname{assumption}{Assumption}{Assumptions}
\crefname{proposition}{Proposition}{Propositions}
\crefname{lemma}{Lemma}{Lemmas}
\crefname{remark}{Remark}{Remarks}
\crefname{section}{Section}{Sections}
\crefname{corollary}{Corollary}{Corollaries}
\crefname{table}{Table}{Tables}

\begin{document}

\title{{{Discrete-Time}} Conewise Linear Systems with Finitely Many Switches}
\author{Jamal Daafouz, J\'er\^ome Loh\'eac, Constantin Mor\u{a}rescu, Romain Postoyan
\thanks{The authors are with Universit\'e de Lorraine, CNRS, CRAN, F-54000 Nancy, France (e-mails: firstname.surname@univ-lorraine.fr). J.~Daafouz is also with Institut Universitaire de France (IUF).}
\thanks{This work is supported by IUF and by the ANR through the projects NICETWEET, ANR-20-CE48-0009, TRECOS, ANR-20-CE40-0009, and OLYMPIA, ANR-23-CE48-0006.} 
}
\maketitle

\begin{abstract} 
We investigate discrete-time conewise linear systems (CLS) for which all the solutions exhibit a finite number of switches. By switches, we mean transitions of a solution from one cone to another. Our interest in this class of CLS comes from the optimization-based control of an insulin infusion model, for which the fact that solutions switch finitely many times appears to be key to establish the global exponential stability of the origin. The stability analysis of this class of CLS greatly simplifies compared to general CLS, as all solutions eventually exhibit linear dynamics. The main challenge is to  characterize CLS satisfying this finite number of switches property. 
We first present general conditions in terms of set intersections for this purpose. To ease the testing of these conditions, we translate them as a non-negativity test of linear forms using Farkas lemma. As a result, the problem reduces to verify the non-negativity of a single solution to an auxiliary linear discrete-time system. Interestingly, this property  differs from the classical non-negativity problem, where any solution to a
system must remain non-negative (component-wise) for any non-negative initial condition, and thus requires novel tools to test it. We finally illustrate the relevance of the presented results on  the optimal insulin infusion problem.
\end{abstract}

\begin{IEEEkeywords}
Conewise linear systems, Lyapunov stability, optimization-based control, insulin infusion, Farkas lemma.
\end{IEEEkeywords}
	
\section{Introduction}
Conewise linear systems (CLS) are dynamical systems for which the state space is partitioned into a finite number of non-overlapping polyhedral cones \cite{Camlibel2006, Shen2010, Wirth2014}. The dynamics within each cone is governed by a linear time-invariant dynamical system called a mode. These systems pose significant challenges due to their piecewise linear nature. 
In particular, it has been shown in \cite{Blondel1999} that the stability analysis of CLS is a NP-hard problem. Hence, algorithms for deciding stability of CLS are inherently inefficient. 
We can mention the converse Lyapunov results in \cite{Wirth2014}, which lead to necessary and sufficient stability conditions. The results in \cite{Wirth2014} state that the origin of a CLS is globally exponentially stable (GES) if and only if it admits a conewise linear Lyapunov function, whose associated conic partition does not coincide with the original system partition in general. 
As a consequence, the method proposed in~\cite{Wirth2014} 
may be undecidable or computationally intractable. We can also mention that CLS are known to be equivalent to a class of linear complementarity systems \cite{Camlibel2006} for which cone-copositive Lyapunov functions may be synthesized, as done recently in \cite{Souaiby2021} in continuous-time. The advantage is the derivation of converse results with polynomial approximations but, again, there is no guarantee of computational tractability in general. To alleviate the computational obstruction of the CLS stability analysis, an alternative approach is to exploit additional properties for classes of CLS. 
This is the approach pursued in this work, where we focus on discrete-time CLS, whose solutions all exhibit a finite number of switches. By switches, we mean the transition of a solution from one cone to another. 

Our motivation to study CLS with finitely many switches comes from the application of the optimization-based control approach of~\cite{Boyd2011a}, namely quadratic control-Lyapunov policy (QCLP), for optimal insulin infusion in \cite{Good2019}. The primary objective is to minimize peak blood glucose level (BGL) caused by a food impulse, while adhering to the constraint that insulin flow must be positive \cite{Good2019}. Under the assumption that 
the response of the meal lasts longer than the response of an insulin impulse, as observed in specific situations (e.g., low glycemic index or high-fat/protein meals), it is proved in \cite{Good2019} that the optimal infusion policy is given by an open-loop policy combining an insulin bolus (applied with the meal) and a specific form of decaying insulin flow thereafter. This strategy exhibits the shortcomings of being open-loop and of  requiring the perfect knowledge of the meals. This justifies the need for alternative efficient, suboptimal, closed-loop feedback policies. 
It appears that the problem can be formalized as a Linear Quadratic Regulator (LQR) problem with positive control inputs, as we show. This constrained optimal LQR problem arises in numerous real-world systems and leads to major methodological challenges \cite{Pachter1980, Heemels1998, Bempo2002, Boyd2011, Loheac2021}. Necessary and sufficient optimality conditions for the continuous-time case are available in \cite{Heemels1998}. 
Nevertheless, the associated numerical algorithm suffers from the curse of dimensionality,  approximations are thus required to determine the stationary infinite horizon optimal feedback, moreover and importantly, no stability guarantees are provided. This is the reason why we opted for QCLP \cite{Boyd2011} instead, which is easy to implement and shows remarkable performances in simulations when applied to the model in \cite{Good2019}; see Figure \ref{diabf1c} in Section \ref{sect:motivating-example}. We show that the obtained closed-loop system can be modelled as a CLS but ensuring its stability remains challenging. Indeed, despite our attempts to apply existing methods for stability analysis of CLS \cite{Wirth2014} or optimization-based control systems, such as \cite{Primbs2001} and its extensions \cite{Li2008,Ahmad2014,Park}, as well as sum-of-squares-based methods like \cite{Korda2002} or the discrete-time counterpart of \cite{Souaiby2021}, we were unable to obtain a stability certificate for the CLS under consideration.

We realized that the above-mentioned motivating example satisfies a distinctive feature: the number of switches of any solution to the closed-loop system is uniformly bounded. This property is extremely useful to investigate stability, as any solution is eventually given by a linear dynamical system. As, for this example, each mode has a Schur state matrix, the global exponential stability of the origin can then be established. Our goal in this work is to formalize and generalize these findings. 

We start by assuming that a given general discrete-time CLS is such that all its solutions exhibit a finite number of switches, and we give a necessary and sufficient condition for the origin to be globally exponentially stable. This stability property relies on the  {assumption} that solutions switch finitely many times, which is non-robust to exogenous disturbances a priori. We might thus deduce that the ensured stability property is not robust. We prove that this is not the case by establishing that, for a CLS, global exponential stability of the origin implies exponential input-to-state stability with respect to a general class of additive disturbances. Afterwards, we focus on the main challenge of this work that is to derive conditions under which all solutions to a CLS exhibit a finite number of switches.  {We are not aware of such results in the literature. We can mention e.g., \cite{shen-pang-siam2005},  which deals with linear complementarity systems with continuous-time dynamics and not discrete-time dynamics as we do, and ensures finitely many switches for any solution  on finite time intervals only, while we are seeking for results over the whole domain of the solutions. There is therefore a need for novel methodological tools that allows establishing that any solution to given discrete-time CLS switch a given maximum number of times.}

 {In this context, we present general conditions in terms of set intersections.} To ease the testing of these conditions, we derive alternative, tractable conditions, which boil down to verifying whether a specific solution to an auxiliary discrete-time system is non-negative thanks to the use of Farkas lemma. These conditions are derived for the case of a partition of the state space made of two cones, only to avoid over complicating the used notation. It is interesting to note that the required non-negativity property differs from the abundant literature on positive systems e.g.,~\cite{Farina2000}, which concentrate on systems for which \emph{all} solutions take non-negative values (component-wise). In our case, a single solution has to be non-negative, we therefore present tailored conditions for this purpose, which have their own interest and which are successfully applied to the optimal insulin infusion problem.

The rest of the paper is organized as follows. The motivating example inspired by \cite{Good2019} and the problem statement are presented in Section~\ref{sect:motivating-example}. Stability results are established in Section~\ref{sect:stability} assuming all the solutions of the considered CLS have a finite number of switches. 
Section \ref{sect:cls-finite-number-switches} provides conditions under which solutions to a CLS exhibit no more than a given maximum number of switches. Section \ref{sect:positivity-condition} is dedicated to the non-negativity analysis of the auxiliary system derived from Farkas Lemma. 
The results are finally applied to examples, including the insulin infusion problem, in Section~\ref{sect:applications}.  {Section~\ref{sect:conclusions} concludes the paper.} 
Some results are postponed to the appendix to avoid breaking the flow of exposition. 

\noindent
{\bf Notations.} $\R$ denotes the set of real numbers, $\R_+$ the set of non-negative real numbers, $\N$ the set of non-negative integers, $\N^\star:=\N\backslash\{0\}$, and $\mathbb{C}$ the set of complex numbers. For real matrices or vectors ($^{\top}$) indicates
transpose. The  identity matrix of the considered set of matrices is denoted $\1$. For any symmetric matrix $X > 0$ ($X \geq 0$) means that $X$ is positive (semi-)definite. Given $a_1,\ldots,a_n\in\R$ with $n\in\N^\star$, we use $\diag(a_1,\ldots,a_n)$ to denote the diagonal matrix, whose diagonal components are $(a_1,\ldots,a_n)$. {For any real square matrix $M$, $\sigma(M)$ denotes its spectrum}.
For any vectors $v$, we write $v\geq 0$ when all its entries are non-negative. The notation $\lfloor s\rfloor$ for $s\in\R$ stands for the integer part of~$s$, and recall that $s-1\le\lfloor s\rfloor$. The interior of a set $S$ is denoted $\inter(S)$, its closure $\cl(S)$ and we use $\mathbb{B}(0,r)$ for the closed ball centered at the origin of radius $r>0$ of the considered Euclidean space. Also, $\card(S)$ is the cardinal of the set~$S$.  Given $n\in\N^\star$,  we say that set $\mathcal{C}\subset\R^n$ is a closed convex cone of $\R^n$ if $\mathcal{C}$ is closed and, for any $x,y\in\mathcal{C}$ and any $a,b\in\R_+$,  $a x+b y\in \mathcal{C}$. Given two sets $A$ and $B$, $A^B$ stands for the set of functions defined {from~$B$ to~$A$}.

\section{Motivation and problem formulation}\label{sect:motivating-example}

We first focus on the LQR problem with scalar positive inputs (Section~\ref{subsect:lqr-positive-inputs}), which covers the optimal insulin infusion problem presented in more details afterwards (Section~\ref{subsect:optimal-insulin-infusion}). After having illustrated the potential of QCLP for the {near-optimal insulin infusion}, we {formalize the stability analysis of the obtained closed-loop system as the stability problem of a CLS (Section~\ref{subsect:ex-cls-system}). There, we also discuss possible, but unfortunately unsuccessful, Lyapunov-based approaches to establish stability properties for the considered example thereby motivating the problem statement (Section~\ref{sect:problem-statement}).} 

\subsection{LQR with positive inputs}\label{subsect:lqr-positive-inputs}

The envisioned optimal insulin infusion problem is modeled as a LQR problem with positive inputs. We thus consider the deterministic discrete-time linear system
\begin{equation}\label{eqdisc}
	x_{t+1} = A x_t + Bu_t, 
\end{equation}
where $x _t =(x_t^1, \ldots, x_t^n)^\top \in \R ^n$ is the state and $u_t \in \R_+$ is the non-negative scalar control input at time $t\in\N$, with $n\in\N^\star$. We denote the solution to (\ref{eqdisc}) initialized at state $x\in\R^{n}$ with input sequence $\mathbf{u}\in\R_+^{\N}$ at time $t\in\N$ as $\phi(t,x,\mathbf{u})$ and $\phi(0,x,\cdot)=x$. The cost function is given, for any $x\in\R^n$ and infinite-length sequence of non-negative inputs $\mathbf{u}=(u_0,u_1,\ldots)\in\R_+^{\N}$, by
\begin{equation}
	\begin{array}{rllll}
		J(x,\mathbf{u}) & := & \displaystyle\sum_{t=0}^{\infty} \ell(\phi(t,x,\mathbf{u}),u_t),
	\end{array}
	\label{eq:cost-J}
\end{equation}
where $\ell(z,v):=z^{\top}Qz + 2z^{\top}Sv + {Rv^2}$ for any $(z,v)\in\R^{n}\times\R_+$, $Q\in\R^{n\times n}$ such that $Q=Q^{\top} \geq 0$, $R\in(0,\infty)$, $S\in\R ^n$ and $\left[\begin{matrix} Q & S \\ S^\top & R \end{matrix}\right]\geq 0$. We know from Bellman equation that the optimal value function associated {with} (\ref{eq:cost-J}), namely $V^{\star}(x):=\min_{\mathbf{u}\in\R_+^{\N}}J(x,\mathbf{u})$ satisfies, for any $x\in\R^{n}$,
\begin{equation}\label{eqValueF}
	V^\star(x) \defeq \min_{u\in\R_+}\Big(\ell(x,u)+V^{\star}(Ax+Bu)\Big)
\end{equation}
based on which  the optimal feedback policy is given by
\begin{equation}\label{eqBell}
	g^\star(x) = \min_{u\in\R_+}\Big(\ell(x,u)+V^{\star}(Ax+Bu)\Big).
\end{equation}
Hence to construct the optimal policy $g^\star$, we need to know~$V^{\star}$, which is very challenging in general because of the constraint that $u$ has to be non-negative; see \cite{Heemels1998} for results in the continuous-time case. Consequently, we exploit an alternative suboptimal policy proposed in \cite{Boyd2011} known as QCLP. The idea is to construct the feedback policy using a known control Lyapunov function $V_{\text{clf}}$ instead of $V^{\star}$ in (\ref{eqBell}). This leads to, for any $x\in\R^{n}$,
\begin{equation}\label{eq:QCLP}
g(x) := \argmin_{u\in\R_+} \big( \ell(x,u) + V_{\text{clf}}(Ax+Bu)\big).
\end{equation}
In this work, and as proposed in \cite{Boyd2011a}, we define $V_{\text{clf}}$ as the optimal value function associated {with} the \emph{unconstrained} LQR problem in the sense that the input can take any value in $\R$, i.e., $V_{\text{clf}}(x):=x^{\top}Px$ for any $x\in\R^{n}$ where $P$ is the unique, real, symmetric, positive definite solution to the Riccati equation $A^{\top}PA-P	 -(A^{\top}PB+S)(B^{\top}PB+R)^{-1}(A^{\top}PB+S)^{\top}+Q=0$, which exists as long as the pair $(A,B)$ is controllable. 

In view of the expressions of $\ell$ and $V_{\text{clf}}$, (\ref{eq:QCLP}) can be re-written~as
\begin{equation}\label{QP99}
\begin{array}{r@{\ }c@{\ }lll}
g(x) & = &  \argmin_{u\in\R_+} {\big((R+B^{\top}PB)(u-Kx)^2\big)} \\
& = & \max\{0,Kx\},\end{array}
\end{equation}
with $K=-(R+B^\top P B)^{-1}(B^\top P A+S^\top)$ the optimal gain for the unconstrained LQR problem. 
This equivalence follows from (\ref{eq:QCLP}) and the equalities $K^{\top}(R+B^{\top}PB)K = A^{\top}PA-P+Q$ and $-(R+B^{\top}PB)K = B^{\top}PA+S^{\top}$.

\subsection{Optimal insulin infusion}\label{subsect:optimal-insulin-infusion}

We apply the approach of Section~\ref{subsect:lqr-positive-inputs} to the problem of optimal insulin infusion of~\cite{Good2019} where
the primary objective is to minimize the peak of the blood glucose level (BGL) resulting from a food impulse while ensuring that insulin flow remains non-negative. 
We reformulate the problem as a LQR problem with a control non-negativity constraint and {apply} the associated QCLP strategy~(\ref{eq:QCLP}). The continuous-time model of~\cite{Good2019} was obtained from clinical trials. Here, we use a sampled-data version with a sampling period $T=5$ min; see Appendix~\ref{appendix:sampled} for more details. As a result, the discrete-time state model is given by (\ref{eqdisc}) with $n=4$, and
\begin{equation}
	\begin{array}{l}\label{eq:matrices-insulin-A-B}
		A= \left[ \begin{array}{rrrc}
				0.8351  & -0.1150 &  -0.0521    &     0 \\
				0.0716  &  0.9954&   -0.0021     &    0 \\
				0.0014   & 0.0390  &  1.0000   &      0 \\
		-0.0082 &  -0.3249 & -16.4423  &  0.9277 \end{array} \right] \\
		B =\left[ \begin{array}{r}
				1.6702\\
				0.1431\\
				0.0029\\
		-0.0163 \end{array} \right].
	\end{array}
\end{equation}
We have selected the weighting matrices $Q,S$ and $R$ as in Appendix~\ref{appendix:sampled}, which leads to \begin{equation}\begin{array}{l}\label{eq:matrices-insulin-K}
		K = \left[\begin{array}{cccc}
				-0.4936 &   -6.9988 & -104.7360  & 1.6626
		\end{array}\right].
	\end{array}
\end{equation}
A solution to the  corresponding  closed-loop system with (\ref{QP99}) is depicted in Figure~\ref{diabf1c} together with the solutions using the optimal open-loop strategy given in~\cite{Good2019} and the zero-input strategy. As it can be seen in Figure~\ref{diabf1c}, the control is effective in reducing the excursion due to an impulse of food ingested ($60$~g) compared to the open-loop behavior with zero input in black dash-dotted line. 
 Our result shows a maximum BGL excursion of $2$ mmol/L. The best feedback policy in \cite[Section~7]{Good2019}, for the corresponding continuous-time model (see Appendix \ref{appendix:sampled}) has a maximum BGL excursion of $1.28$ mmol/L. However, the associated controller assumes the BGL can be exactly differentiated, which is unfortunately unrealistic.
To make the strategy realistic, the bandwidth of these differentiators has been restricted in \cite{Good2019}, and the results show larger peaks ($5.71$ mmol/L).

\begin{figure}[ht]
	\begin{center}
		\includegraphics[trim=150 0 150 0,clip,width=\the\columnwidth]{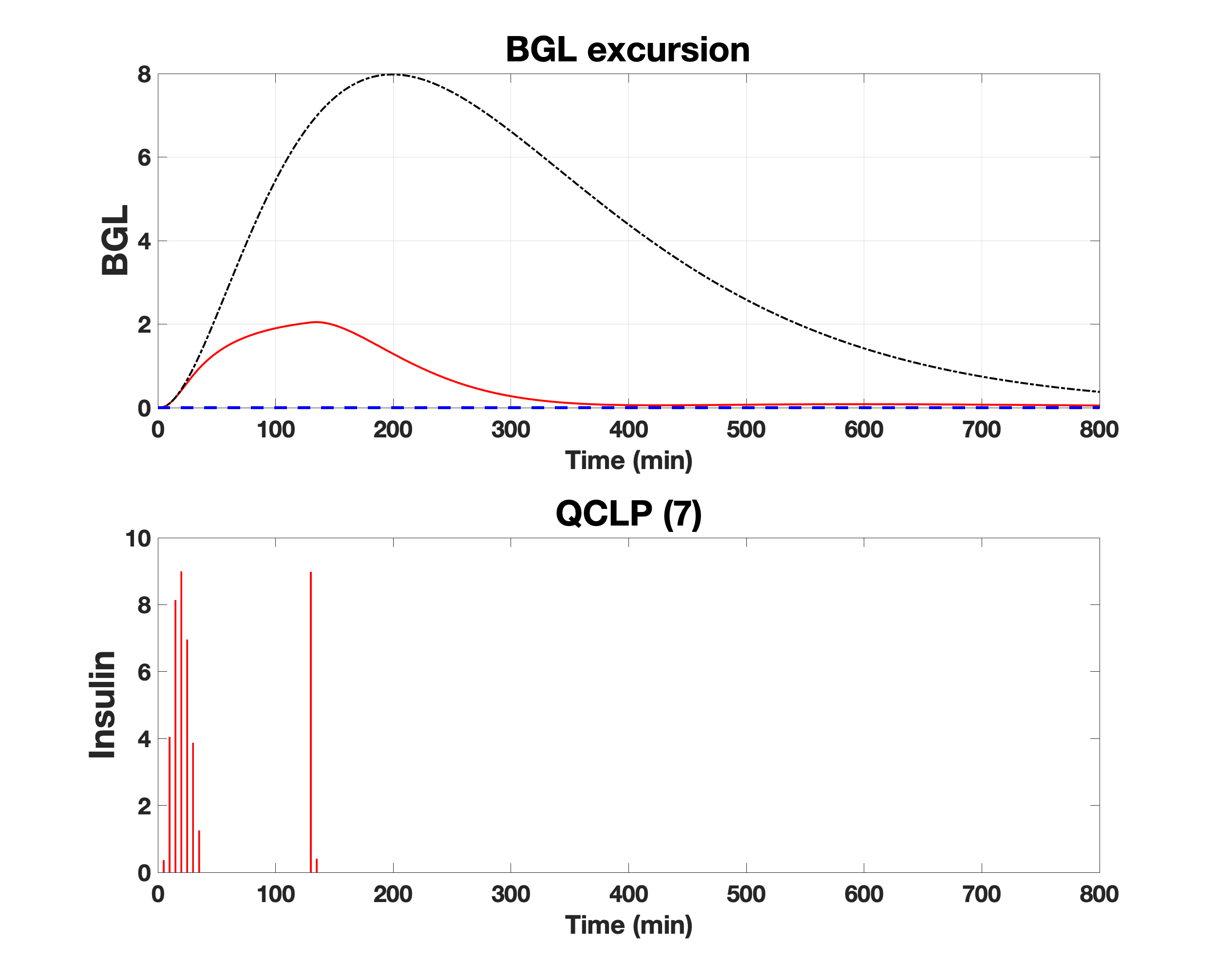}
    \end{center}
	\caption{\label{diabf1c}
        Top: glucose evolution (open-loop with zero input in black dash-dotted line, closed-loop with $K$ in red solid line, optimal open-loop of \cite{Good2019}  in blue dashed line). Bottom: Insulin flow (closed-loop with $K$ in (\ref{eq:matrices-insulin-K})).
    }
\end{figure}

QCLP (\ref{QP99}) is thus  promising  for this application. However, the stability of the corresponding closed-loop system is an open question, as we now explain. 

\subsection{{Limitations of existing Lyapunov-based techniques}}\label{subsect:ex-cls-system}

{System (\ref{eqdisc}) in closed-loop with (\ref{QP99}) is a CLS given by}
\begin{equation}\label{sysPWQ}
	x_{t+1} = \begin{cases}
		A_1x_t  & x_t \in  \mathcal{C}_1,  \\
		A_2x_t  & x_t \in  \mathcal{C}_2,
	\end{cases}
\end{equation}
with $A_1 = A+BK$ and $A_2=A$ and
\begin{equation}\label{eqcvx12}
\begin{array}{r@{\>}l}
\mathcal{C}_1&=\left\{ x\in\R^n\ :\ Kx\geq0 \right\},\\
\mathcal{C}_2&=\left\{ x\in\R^n\ :\ -Kx\geq0 \right\}=\cl\left(\R^n\backslash\mathcal{C}_1\right),
\end{array}
\end{equation}
where ${K \in \R^{1\times n}}$ is the gain associated {with} the optimal unconstrained LQR problem. {It appears that certifying the stability of the origin for system (\ref{sysPWQ}) is challenging \cite{Blondel1999}; even when $A_1$ and~$A_2$ are Schur as in Section~\ref{subsect:optimal-insulin-infusion}.}

{Interestingly, this stability problem can be formulated in other frameworks such as optimization-based systems due to~(\ref{eq:QCLP}), piecewise linear systems to which system (\ref{sysPWQ}) belongs to, or linear complementarity systems\footnote{This can be seen by noticing that the necessary and sufficient optimality Karush-Kuhn-Tucker (KKT) conditions of the quadratic problem~(\ref{eq:QCLP}) write $ (R+B^{\top}PB)v -(R+B^{\top}PB)K x_t - \lambda  = 0$, $v \geq 0$, $\lambda  \geq 0$, $\lambda v=0$, 
where $\lambda \in \R$ is a Lagrange multiplier. As a result, the closed-loop system (\ref{eqdisc})-(\ref{eq:QCLP}) can also be written as a discrete-time LCS $x_{t+1} = Ax_t + Bu_t$, $y_t = Cx_t + Du_t$,
$0 \leq u_t \perp y_t \geq 0$, 
where $C = -(R+B^{\top}PB)K$ and $D = R+B^{\top}PB$. } (LCS) {a well-established class of nonsmooth systems~\cite{Brogliato2020}. This connection is significant, as it situates our problem within a broader theoretical framework. However, while the LCS theory provides a useful perspective, solving the problem at hand within this framework is nontrivial and cannot be fully addressed by existing methods in the literature as far as we know. This opens an interesting avenue for future research, requiring the development of new techniques to handle the specific challenges posed by the considered system.

To elaborate more on the fact that none of the existing Lyapunov-based stability tools of the literature certifies the stability of the system in Section~\mbox{\ref{subsect:optimal-insulin-infusion}},} we introduce below a novel comparison between three distinct classes of Lyapunov function candidates commonly encountered in the above-mentioned fields. The first one is found in the optimization-based framework \cite{Primbs2001,Postoyan-et-al-tac(optimal),Grimm-et-al-tac2005}, and corresponds to the value function associated {with}~(\ref{eq:QCLP}), i.e., for any $x\in\R^{n}$,
\begin{equation}\label{eqLyapVk}
\begin{array}{l}
V(x) := \displaystyle\min_{u\in\R_+}\Big(x^{\top} Qx+2x^{\top} Su+u^{\top} Ru
 \\[-2mm] \ \hskip 3.0cm
 +(Ax+Bu)^{\top} P(Ax+Bu)\Big).
\end{array}
\end{equation}
The second Lyapunov function candidate is the quadratic form inherited from the LCS literature \cite{Camlibel2007}, i.e., for any $x\in\R^{n_x}$ and $u=g(x)$ as in (\ref{eq:QCLP}),
\begin{equation}\label{eqLyapLCP}
	V_{\text{quad}}(x) : =  \left[ \begin{array}{cc} x^{\top}  & u^{\top}   \end{array}\right] \underbrace{\left[ \begin{array}{cc}  X_1 & X_2 \\ X_2^{\top}  & X_3 \end{array}\right]}_{\cal X}\left[ \begin{array}{c} x \\ u  \end{array}\right]  
\end{equation}
where $X_1=X_1^\top  \in \R ^{n \times n}$, $X_2 \in \R ^{n \times m}$, and $X_3=X_3^\top  \in \R ^{m \times m}$ and ${\mathcal X}$ strictly copositive \cite{Camlibel2007}.
The third one is the piecewise quadratic function, as in \cite{Feng2002}, that is for any $x\in\R^{n}$,
\begin{equation}\label{eqLyapPWQ}
V_{\text{PWQ}}(x) =:\begin{cases}
x^{\top}  P_1 x, & \text{if } x \in {\cal C}_1, \\
x^{\top}  P_2 x, & \text{if } x \in {\cal C}_2,
\end{cases}
\end{equation}
where $P_1=P_1^\top  \in \R ^{n \times n}$ and $P_2=P_2^\top  \in \R ^{n \times n}$.
This function is also a conewise quadratic function with the same partition as~(\ref{sysPWQ}). 

The relationship between these three Lyapunov function candidates unfolds as follows. Function $V$ is a particular instance of $V_{\text{quad}}$. Specifically, if we set, in~(\ref{eqLyapLCP}), $X_1 = A^{\top} PA+Q$, $X_2 = A^{\top} PB+S$, and $X_3 = B^{\top} PB+R$, then we obtain $V_{\text{quad}}=V$. Function $V_{\text{PWQ}}$ is the most general form, encompassing both (\ref{eqLyapVk}) and (\ref{eqLyapLCP}), as it accommodates $V_{\text{quad}}$ (which is more general than $V$) by setting $P_1= X_1+K^{\top} X_2^{\top} +X_2K+K^{\top} X_3K$ and $P_2 = X_1$, leading to $V_{\text{PWQ}}=V_{\text{quad}}$. Having identified the most general Lyapunov function candidate among these three options for system (\ref{sysPWQ}), LMI conditions under which the stability of the origin of~(\ref{sysPWQ}) would be guaranteed using (\ref{eqLyapPWQ}) are provided in Appendix~\ref{appendix:PWQ}. Regrettably, these conditions are infeasible for the system discussed in Section \ref{subsect:optimal-insulin-infusion}.  Furthermore, all our attempts to ensure the stability of this system based on other existing tools from the literature \cite{Primbs2001,Wirth2014,Souaiby2021,Li2008,Ahmad2014,Park} failed.} 

Now, for all the tested values of $Q$, $R$ and~$S$ in (\ref{eq:cost-J}), all the solutions of the corresponding system (\ref{sysPWQ}) have the distinctive feature to switch a maximum of $4$ times, which, once analytically established, allows concluding on the stability of the origin for the considered system, see Section \ref{subsect:applications-insulin}. In the following, we formalize these findings for general CLS. 

\subsection{Problem statement}\label{sect:problem-statement}

Motivated by the above developments, we consider in this paper general discrete-time CLS of the form
\begin{equation}\label{eqCLS}
	x_{t+1} = A_i x_t \quad \quad x_t \in {\cal C}_i, \quad i \in \{1,\dots,m\},
\end{equation}
where $m\in\N^\star$ is the number of non-empty closed convex cones $\mathcal{C}_i$, $i\in\{1,\ldots,m\}$, which partition the state space~$\R^n$, i.e., 
\begin{itemize}
	\item ${\cal C}_1\cup\dots\cup{\cal C}_m = \R ^n$,
	\item $\inter({\cal C}_i ) \cap \inter({\cal C}_j) = \emptyset$ for all $i\neq j \in\{1,\dots,m\}$.
\end{itemize}

{\begin{assumption}\label{ass_cont}
	$A_i \xi = A_j \xi$ for any vector $\xi \in {\cal C}_i \cap {\cal C}_j$.
\end{assumption}
As in \cite{Wirth2014}, the continuity condition of Assumption~\ref{ass_cont} ensures consistent behavior across the switching surfaces (the cone boundaries) and guarantees the uniqueness of solutions for system (\ref{eqCLS}) for any initial condition. This assumption is crucial because, in discrete-time systems, it is possible for trajectories to converge in finite time towards sliding surfaces, as discussed in \cite{Brogliato2021}. The continuity assumption we make here prevents the occurrence of sliding modes and attractive surfaces, which would otherwise complicate the system dynamics.
\begin{assumption}\label{ass:regular}
    Matrices $A_1,\ldots,A_m$ are invertible.
\end{assumption}
The full-rank condition imposed by Assumption~\ref{ass:regular} is quite classical for discrete-time systems.
In our study, this assumption will allow us to avoid additional technical difficulties.}

For the sake of convenience and with some slight abuse of notation, we use $\phi$ to denote solutions to (\ref{eqCLS}), so that for initial condition $x\in\R^n$, the corresponding solution to (\ref{eqCLS}) at time $t\in\N$ is denoted $\phi(t,x)$.
We refer to the dynamics of~(\ref{eqCLS}) for a given $i\in\{1,\ldots,m\}$ as a mode.
{%
We next formalize what is meant by a switch and the maximum number of switches of a solution to (\ref{eqCLS}).
Intuitively, one would say that, given a solution to (\ref{eqCLS}), a switching time occurs when the solution leaves a set $\mathcal{C}_i$.
However, since our sets are not strictly disjoint, this simple definition is not suitable.
For instance, with $m=3$, if we have a sequence such that $x\in\mathcal{C}_1\setminus\left( \mathcal{C}_2\cup\mathcal{C}_3 \right)$, $\phi(1,x)\in\left( \mathcal{C}_1\cap\mathcal{C}_2 \right)\setminus\mathcal{C}_3$, $\phi(2,x)\in \mathcal{C}_1\cap\mathcal{C}_2 \cap\mathcal{C}_3$, $\phi(3,x)\in\left( \mathcal{C}_2\cap\mathcal{C}_3 \right)\setminus\mathcal{C}_1$ and $\phi(t,x)\in\mathcal{C}_3\setminus\left( \mathcal{C}_1\cup\mathcal{C}_2 \right)$ for every $t\ge4$, we would say that switching times are~$3$ and~$4$, since $\mathcal{C}_1$ (respectively $\mathcal{C}_2$) is left at time $t=3$ (respectively $t=4$).
However, for this particular case, the \textsl{correct} switching time is $t=3$, since $\phi(t,x)\in\mathcal{C}_1$ for $t\in\{0,1,2\}$ and $\phi(t,x)\in\mathcal{C}_3$ for $t\ge3$.
Thus, in order to define the switching times and the number of switches, we need to explore all possible sequences of modes.
To this end, in the spitit of \cite[Section~5]{Wirth2014}, for every $x\in\R^n$, we define the possible sequences of modes
\[\mathcal{I}(x)=\left\{ (i_t)_{t\in\N}\in\{1,\ldots,m\}^\N\ :\ \forall t\in\N,\ \phi(t,x)\in\mathcal{C}_{i_t} \right\}.\]
Let us now define for $(i_t)_{t\in\N}\in\{1,\ldots,m\}^\N$ the jumping times:
\[\mathcal{T}\left( (i_t)_t \right)=\left\{ t\in\N^*\ :\ i_t-i_{t-1}\neq0 \right\}\subset\N^*\]
and the number of jumps
\[\jmath\left( (i_t)_t \right)=\card\left( \mathcal{T}\left( i_t)_t \right) \right)\in\N\cup\{\infty\}.\]
We are now in position to define the number of switches and switching times associated with a solution of~\eqref{eqCLS}.
\begin{definition}[Switching times and number of switches]\label{def-switch}
	Given an initial condition $x\in\R^{n}$, the \textsl{number of switches} of the solution $\phi(\cdot,x)$ to~\eqref{eqCLS} is given by
	\[\varsigma(x)=\min\left\{ \jmath\left( (i_t)_{t\in\N} \right)\ :\ (i_t)_{t\in\N}\in \mathcal{I}(x) \right\}\in\N\cup\{\infty\}.\]
	Corresponding \textsl{switching times} are the elements of $\mathcal{T}\left( (i^*_t)_t \right)$, where $(i^*_t)_t\in \mathcal{I}(x)$ is a sequence such that $\jmath\left( (i^*_t)_t \right)=\varsigma(x)$.
\end{definition}
Note that, in the above definition, $\varsigma(x)$ and $(i^*_t)_t$ are well-defined since we take the minimum of a non-negative function having integer values. 
}
In view of Definition~\ref{def-switch}, a solution $\phi(\cdot,x)$ for some $x\in\R^n$ exhibits a finite number of switches if and only if $\varsigma(x)<\infty$.

The main objective of this work is to derive conditions under which all solutions to system (\ref{eqCLS}) exhibits a finite number of switches, see Section~\ref{sect:cls-finite-number-switches}.
This is motivated by the fact that  the stability analysis simplifies for this class of CLS, as we show in the next section.
{In \cref{sect:cls-finite-number-switches}, to ensure that $\varsigma(x)$ is uniformly bounded with respect to $x\in\R^n\setminus\{0\}$, we will ensure that $\max\left\{ \jmath\left( (i_t)_{t\in\N} \right)\ :\ (i_t)_{t\in\N}\in \mathcal{I}(x) \right\}$ is uniformly bounded with respect to $x\in\R^n\setminus\{0\}$.}

\section{Stability results}\label{sect:stability}

We first provide a necessary and sufficient condition for $x=0$ to be GES for system (\ref{eqCLS}), again, assuming all its solutions switch a finite number of times{; conditions to ensure this property are provided in Section \ref{sect:cls-finite-number-switches}}. We then show that this stability property is robust, in the sense that an input-to-state stability property holds when (\ref{eqCLS}) is perturbed by exogenous disturbances. 

\subsection{Global exponential stability}\label{subsect:ges}

We define next the set $\mathcal{F}$, which characterizes the region of the state space where solutions to (\ref{eqCLS}) stop switching,
\begin{equation}\label{eq:F}
	\mathcal{F}:=\left\{x\in\R^n\ :\ \exists i\in\{1,\ldots,m\}{,}\,\, \forall t\in\N\,\,  A_i^t x \in\mathcal{C}_i\right\}.
\end{equation}
In other words, $\mathcal{F}=\{x\in\R^n\ :\ \varsigma(x)=0\}$ with the notation of Definition~\ref{def-switch}.
Obviously, $0\in\mathcal{F}$ and $\mathcal{F}$ is forward invariant for system (\ref{eqCLS}).
The next theorem gives a necessary and sufficient condition for the origin of system (\ref{eqCLS}) to be GES, i.e., there exist $c_1\geq 1$ and $c_2>0$ such that for any $x\in\R^{n}$, $|\phi(t,x)|\leq c_1e^{-c_2 t}|x|$ for any $t\in\N$, when we know that any solution {exhibits} a finite number of switches.
Again, conditions to ensure the latter property are provided in Section \ref{sect:cls-finite-number-switches}.

\begin{theorem}\label{thm:stability} Consider system (\ref{eqCLS}) and suppose that any solution exhibits a finite number of switches, i.e., $\varsigma(x)<\infty$ for any $x\in\R^n$ with $\varsigma$ as in Definition~\ref{def-switch}. Then the origin is GES if and only if the origin of the restriction of (\ref{eqCLS}) to $\mathcal{F}$, namely,
\begin{equation}\label{eqCLS-constrained}
	x_{t+1} = A_i x_t \quad \quad x_0 \in {\cal C}_i\cap \mathcal{F}, \quad i \in \{1,\dots,m\},
\end{equation}
is GES.
\end{theorem}
\begin{proof}
	We first suppose the origin of (\ref{eqCLS}) is GES.
    As $\mathcal{F}$ is forward invariant, it follows that the origin of (\ref{eqCLS-constrained}) is GES. 

	Suppose now that the origin of (\ref{eqCLS-constrained}) is GES.
    Let $\phi$ be a solution to~(\ref{eqCLS}) initialized at $x\in\R^n$.
    Since $\phi(\cdot,x)$ exhibits a finite number of switches over $\N$, {for $t_0=t_0(x)$ (the last switching time),} we have $\phi(t_0,x)\in\mathcal{F}$.
    As a result, $\phi(t,x)\in\mathcal{F}$ for any $t\geq t_0$ as $\mathcal{F}$ is forward invariant.
    This implies that $\phi(t,x)\to 0$ as $t\to \infty$ as the origin is GES for system~(\ref{eqCLS-constrained}).
    Since $x$ has been arbitrarily chosen, we have proved that any solution to~(\ref{eqCLS}) asymptotically converges to the origin, i.e., the origin is globally attractive for system~(\ref{eqCLS}).
    As item~\ref{it:thm:attractivity-ges:1} of Theorem~\ref{thm:attractivity-ges} in Appendix~\ref{appendix:attractivity-ges} holds, it yields that the origin is GES for system~(\ref{eqCLS}).
\end{proof}

To apply Theorem~\ref{thm:attractivity-ges} we use two ingredients: that the origin of (\ref{eqCLS-constrained}) is GES, and that any solution enters in~$\mathcal{F}$.
Although not necessary, assuming that any solution to~(\ref{eqCLS}) exhibits a finite number of switches guarantees that any solution enters \cal{F}.
We note that $\mathcal{C}_i\cap\mathcal{F}$, which appears in~(\ref{eqCLS-constrained}), is the largest forward invariant set within cone $\mathcal{C}_i$ for the dynamics $x_{t+1}=A_i x_t$.
When we know in which cone(s) the solutions eventually enter and remain for all future times, this boils down to investigating the spectrum of the state matrices associated {with} these cones.
When all the matrices~$A_i$, $i\in\{1,\ldots,m\}$, are Schur as in the examples of Section~\ref{sect:applications}, the fact that the origin of~(\ref{eqCLS-constrained}) is GES directly follows as solutions to~(\ref{eqCLS-constrained}) exhibit no switches. 
Now ensuring the condition that all solutions exhibit a finite number of switches, requires the development of novel tools, which are presented in Section~\ref{sect:cls-finite-number-switches}.


\begin{remark}
	A necessary condition for the origin of system~(\ref{eqCLS-constrained}) to be GES is that for every $i\in\{1,\dots,m\}$, if $\lambda\in\sigma(A_i)\cap\R$ is such that there exists an eigenvector $v\in\mathcal{C}_i$ of $A_i$ associated {with} $\lambda$, then $\lambda\in[0,1)$.
    This comes from the fact that, with these conditions, the solution of~\eqref{eqCLS} initialized at $v$, $\phi(t,v)=\lambda^t v$, belongs to $\mathcal{C}_i$ for every $t\in\N$.
    In particular, we have $\R_+ v\subseteq \mathcal{F}$.
\end{remark}

\subsection{Exponential input-to-state stability}

The property assumed in Theorem \ref{thm:stability} that any solution to~(\ref{eqCLS}) switches a finite number of times is non-robust a priori, in the sense that an arbitrarily small exogenous disturbance may destroy it. {It is therefore essential  to analyze the robustness of the stability result in Theorem~\ref{thm:stability}.} 
This is formalized in the next theorem, which we could not find in the literature, although similar statements have been established for discrete-time systems but with different homogeneity degrees, see \cite[Theorem~10]{sanchez-et-al-aut20}, or for general continuous-time homogeneous systems, see e.g.~\cite{ryan-scl95}.

We consider for this purpose system (\ref{eqCLS}) perturbed by exogenous disturbance as follows 
\begin{equation}\label{eqCLS-perturbed}
	x_{t+1} \in\left\{ A_i x_t + w_t \,:\, i\in\{1,\ldots,m\},\,\, x_t\in \mathcal{C}_i,\,\,w_t\in E(x_t)v_t\right\},
\end{equation}
where $v_t\in\R^{n_v}$, with $n_v\in\N^\star$, is the disturbance at time~$t$ and $E$ is a set-valued map from $\R^n$ to $\R^{n\times n_v}$ such that: (i) $E(\lambda x)=E(x)$ for any $x\in\R^{n}$ and $\lambda >0$; (ii) there exists $m_E\geq 0$ such that for any $x\in\R^{n}$ and  $z\in E(x)$, $|z|\leq m_E$. Set-valued map $E$ covers as a special case the situation where  $E(x)=\{E_i\ :\ i\in\{1,\ldots,m\},\,\, x\in \mathcal{C}_i\}$ for some constant matrices $E_i$.  On the other hand, matrices $A_i$ and sets $\mathcal{C}_i$, $i\in\{1,\ldots,m\}$, are as  in (\ref{eqCLS}). System (\ref{eqCLS-perturbed}) is a difference inclusion, as its right-hand side is set-valued. For the sake of convenience and like in Section~\mbox{\ref{subsect:lqr-positive-inputs}}, we denote a solution to (\ref{eqCLS-perturbed}) initialized at $x\in\R^{n}$ with disturbance sequence $\mathbf{v}=(v_0,v_1,\ldots)\in(\R^{n_v})^\N$ at time $t\in\N$ as $\phi(t,x,\mathbf{v})$. We have the next robustness result. 

\begin{theorem}\label{thm:iss}
	Suppose the origin is GES for system (\ref{eqCLS}), then system (\ref{eqCLS-perturbed}) is exponentially input-to-state stable, in particular there exist $c_1\geq 1$, $c_2>0$ and $c_3\ge0$ such that for any $x\in\R^{n}$ and any $\mathbf{v}\in(\R^{n_v})^\N$, any solution $\phi$ satisfies $|\phi(t,x,\mathbf{v})|\leq c_1 e^{-c_2 t}|x|+c_3\sup_{t'\in\{1,\ldots,t\}}|v_{t'}|$.
\end{theorem}
\begin{proof}
	We first apply \cite[Theorem 2]{tuna-teel-cdc04} to obtain a suitable homogeneous Lyapunov function for system (\ref{eqCLS}). With the notation of \cite{tuna-teel-cdc04}, we take $\omega(\cdot)=|\cdot|$ and $\mathcal{D}=\R^n$. Then, Assumption~1 in \cite{tuna-teel-cdc04} holds by \cite[Theorems 6.30 and 7.21]{Goebel-Sanfelice-Teel-book}, as $x=0$ is GES for system (\ref{eqCLS}) and the vector field in (\ref{eqCLS}) is continuous\footnote{Recall that $A_i x=A_j x$ for any $x\in\mathcal{C}_i\cap \mathcal{C}_j$.}. Assumption 2 in \cite{tuna-teel-cdc04} also holds, as the vector field in (\ref{eqCLS}) is homogeneous of degree $0$. The last condition to check is \cite[Assumption 3]{tuna-teel-cdc04}, which is verified with $G_\lambda=\lambda \1$ and $d=1$. Consequently, by \cite[Theorem 2]{tuna-teel-cdc04}, there exist $V:\R^n\to\R_{+}$ continuous on $\R^n$, smooth on $\R^{n}\backslash\{0\}$, such that $V(\lambda x)=\lambda V(x)$ for any $\lambda >0$ and $x\in\R^n$, as well as $\underline\alpha,\overline\alpha>0$ and $\mu\in(0,1)$ such that, for any $x\in\R^n$ and any $i\in\{1,\ldots,m\}$ such that $x\in\mathcal{C}_i$,
	\begin{equation}
		\begin{array}{r@{\ }c@{\ }c@{\ }c@{\ }l}
			\underline\alpha|x| & \leq   &  V(x) & \leq & \overline\alpha|x|\\[2mm]
			V(A_i x) & \leq & \mu V(x).
		\end{array}
		\label{eq:proof-iss-lyap}
	\end{equation}
	Let $(x,v)\in\R^{n}\times\R^{n_v}$, $w\in E(x)v$, and any $i\in\{1,\ldots,m\}$ such that $x\in \mathcal{C}_i$.
	In view of (\ref{eq:proof-iss-lyap}) and \cref{res:HomLip} stated below, there exists $L\ge0$ (independent of $x$, $v$ and thus $w$) such that
	\begin{multline}
		V(A_ix+w) \le V(A_i x) + |V(A_i x+w)-V(A_i x)|\\
        \leq \mu V(x) + L|w|\le \mu V(x)+Lm_E|v|.
		\label{eq:proof-issue-first-ineq}
	\end{multline}
	Let $\mathbf{v}\in(\R^{n_v})^\N$, we derive from \eqref{eq:proof-issue-first-ineq} that for any $t\in\N$, any solution $\phi$ to (\ref{eqCLS-perturbed}) {satisfies}
	\begin{equation*}
		V(\phi(t,x,\mathbf{v})) \leq \mu^t\, V(x) + \frac{Lm_E}{1-\mu}\sup_{t'\in\{1,\ldots,t\}}|v_{t'}|.
	\end{equation*}
	We deduce the desired property with $c_1=\overline\alpha/\underline\alpha$, $c_2=-\ln(\mu)$ and $c_3=Lm_E/(\underline\alpha(1-\mu))$ by invoking the first line in~(\ref{eq:proof-iss-lyap}).
\end{proof}
We have used the next lemma in the proof of \cref{thm:iss}.
\begin{lemma}\label{res:HomLip}
	{%
	Let $V:\R^n\to\R$ be such that $V(\lambda x)=\lambda V(x)$ for every $\lambda>0$ and every $x\in\R^n$.
    $V$ is globally  Lipschitz if and only if $V$ is Lipschitz on the unit sphere of $\R^n$.
	}%
\end{lemma}
{A consequence of Lemma \ref{res:HomLip}  is that, if $V$ is $C^1$ on $\R^n\setminus\{0\}$, then it is globally  Lipschitz on~$\R^n$.}
\begin{proof}
We set $D=\mathbb{B}(0,1)$ and $\partial D=D{\backslash } \inter(\mathbb{B}(0,1))$, the unit sphere of $\R^n$.
    Obviously, if $V$ is Lipschitz on $\R^n$, then $V$ is Lipschitz on~$\partial D$. Reciprocally, we assume that $V$ is Lipschitz on $\partial D$.
	It is enough to prove that $V$ is Lipschitz on~$D$.
	In fact, by homogeneity, if $V$ is Lipschitz on $D$, then $V$ is (globally) Lipschitz on~$\R^n$.
	Assume by contradiction that $V$ is not Lipschitz on $D$.
	Then, for every $k\in\N$, there exists $x_k,y_k\in D$ such that $|V(x_k)-V(y_k)|>k|x_k-y_k|$.
	Obviously, we have $x_k\neq y_k$, and without loss of generality, we can assume that $|x_k|{>}|y_k|$.
	Using the homogeneity of $V$, we can also assume without loss of generality that $x_k\in\partial D$.
	In addition, since $k|x_k-y_k|<|V(x_k)-V(y_k)|\le 2\max_{\xi\in\partial D}|V(\xi)|$, we get that $\lim_{k\to\infty}|y_k|=1$, and in particular, for $k$ large enough, we have $|y_k|\neq0$.
	We then define $z_k=y_k/|y_k|\in\partial D$, and we have
	\begin{align*}
	   |V(x_k)-V(z_k)|
	   & \ge |V(x_k)-V(y_k)|-|V(y_k)-V(z_k)|\\
	   & = |V(x_k)-V(y_k)|-(1-|y_k|)|V(z_k)|\\
	   & > k|x_k-y_k|-(1-|y_k|)\max_{\partial D}|V|.
	\end{align*}
	Since $V$ is Lipschitz on $\partial D$, there {exists} $L\in\R_+$ such that $|V(x_k)-V(z_k)|\le L|x_k-z_k|\le L\left( |x_k-y_k|+(1-|y_k|) \right)$.
	We thus have,
	$$k|x_k-y_k|<L|x_k-y_k|+\left( L+\max_{\partial D}|V| \right)(1-|y_k|).$$
	But, $1-|y_k|=|x_k|-|y_k|\le |x_k-y_k|$, and finally, for every $k\in\N$, we should have $k<2L+\max_{\partial D}|V|$ (recall that $x_k\neq y_k$) which is a contradiction.
\end{proof}

\section{CLS with a finite number of switches}\label{sect:cls-finite-number-switches}

The objective of this section is to provide conditions under which all the solutions to (\ref{eqCLS}) exhibit a finite number of switches, as required by \cref{thm:stability}.
We will actually focus on a stronger property, that is that there exists $p\in\N$ such that $\varsigma(x)\leq p$ for any $x\in\R^n$, where we recall that $\varsigma(x)$ is the number of switches exhibited by the solution of~\eqref{eqCLS} initialized at~$x$, see Definition~\ref{def-switch}. To this end, we first give a sufficient condition based on the emptiness of the intersection of a finite number of sets, see~\cref{pg:Helly}.
To ease the testing of this condition, we translate it in \cref{sect:positivity-condition} in terms of checking the non-negativity of linear forms, by exploiting Farkas lemma. As a result, the problem reduces to verify the non-negativity of a single solution to an auxiliary linear discrete-time system. 

In \cref{pg:Helly}, we consider the general case $m\ge2$, while for the sake of clarity and readability only, we focus on the case $m=2$ in \cref{subsect:m=2}.
Note however that the methodology described in \cref{subsect:m=2} is easily extendable to the general case.

\subsection{A sufficient condition}\label{pg:Helly}

Given $p>0$ time instants $t_1,\dots,t_p>0$ and $t\ge0$, and given indexes $i_1,\dots,i_{p+1}\in\{1,\dots,m\}$ with $i_{k+1}\neq i_{k}$ for $k\in\{1,\dots,p\}$, we define the set $\Sigma_{\subalign{&t_1,\dots,t_p,t\\&i_1,\dots,i_p,i_{p+1}}}$ of initial conditions~$x_0$ such that the solution to~\eqref{eqCLS} satisfies $x_s=A_{i_1}^sx_0\in\mathcal{C}_{i_1}$ for $0\le s<t_1$, $x_{t_1+s}=A_{i_2}^sA_{i_1}^{t_1}x_0\in\mathcal{C}_{i_2}$ for $0\le s<t_2$, etc, and finally $x_{t_1+\dots+t_p+s}=A_{i_{p+1}}^sA_{i_p}^{t_p}\ldots A_{i_1}^{t_1}x_0\in\mathcal{C}_{i_{p+1}}$ for $0\le s\le t$.
The first step is to give a condition under which any solution to~\eqref{eqCLS} initialized in $\Sigma_{\subalign{&t_1,\dots,t_p,t\\&i_1,\dots,i_p,i_{p+1}}}$ admits {no more than $p$ switches}.
{In other words, using notations of \cref{sect:problem-statement}, we will give conditions  ensuring that if $(k_t)_{t\in\N}\in \mathcal{I}(x_0)$ is such that $\left\{ t_1,t_1+t_2,\ldots,t_1+\cdots+t_p \right\}\subset \mathcal{T}\left( (k_t)_t \right)\cap \{1,\ldots,t_1+\cdots+t_p\}$ and $k_{t_1-1}=i_1,\, \ldots\, ,k_{t_1+\cdots+t_p-1}=i_p$ and $k_{t_1+\cdots+t_p+s}=i_{p+1}$ for every $s\in\{0,\ldots,t\}$, then we have $\jmath\left( (k_t)_t \right)=p$.
	This in particular ensures that $\sigma(x_0)\le p$, according to \cref{def-switch}.}

Writing down the definition of $\Sigma_{\subalign{&t_1,\dots,t_p,t\\&i_1,\dots,i_p,i_{p+1}}}$, we obtain that this set is given by the following intersections
\begin{equation}\label{eq:set-Sigma}
	\Sigma_{\subalign{&t_1,\dots,t_p,t\\&i_1,\dots,i_p,i_{p+1}}}
	= \bigcap_{\tau=0}^t \left(S_{\subalign{&t_1,\dots,t_p,\tau\\&i_1,\dots,i_p,i_{p+1}}}\right)
	\cap \left(\bigcap_{k=1}^p\bigcap_{\tau=0}^{t_k-1} S_{\subalign{&t_1,\dots,t_{k-1},\tau\\&i_1,\dots,i_{k-1},i_k}}\right),
\end{equation}
where{, for every $k\in\N$, $t_1,\cdots,t_k\in\N^*$, $\tau\in\N$ and $i_1,\cdots,i_k,i\in\{1,\dots,m\}$ we have set}
\begin{equation}\label{eq:set-S}
	S_{\subalign{&t_1,\dots,t_{k},\tau\\&i_1,\dots,i_{k},i}}:=\left\{ x_{{0}}\in\R^n\ \,:\,\ A^{{\tau}}_iA_{i_k}^{t_k}\dots A_{i_1}^{t_1}x_{{0}}\in \mathcal{C}_{{i}} \right\}.
\end{equation}
For notation convenience, in~\eqref{eq:set-Sigma}, for $k=1$ (resp. $k=2$), $S_{\subalign{&t_1,\dots,t_{k-1},\tau\\&i_1,\dots,i_{k-1},i_k}}$ is identified with $S_{\subalign{&\tau\\&i_1}}:=\big\{ x_{{0}}\in\R^n\ \,:\,\ A^\tau_{i_1}x_{{0}}\in \mathcal{C}_{i_1} \big\}$ (resp. $S_{\subalign{&t_1,\tau\\&i_1,i_2}}=\big\{ x_{{0}}\in\R^n\ \,:\,\ A^\tau_{i_2}A_{i_1}^{t_1}x_{{0}}\in \mathcal{C}_{i_2} \big\}$).

To give a better intuition on the meaning of the sets $\Sigma_{\subalign{&t_1,\dots,t_p,t\\&i_1,\dots,i_p,i_{p+1}}}$ and $S_{\subalign{&t_1,\dots,t_{k},\tau\\&i_1,\dots,i_{k},i}}$, we consider a solution of \eqref{eqCLS} that is initialized at some $x_0\in\mathcal{C}_1$, stays in $\mathcal{C}_1$ for times $t\in\{0,1\}$,  enters in $\mathcal{C}_2$ in time $t=2$ and stays in $\mathcal{C}_2$ at time $t=3$.
We then have:
\begin{itemize}[nosep]
	\item $x_0\in\mathcal{C}_1=S_{\subalign{&0\\&1}}$;
	\item $x_1=A_1x_0\in\mathcal{C}_1$, i.e., $x_0\in S_{\subalign{&1\\&1}}=\left\{ x_0\in\R^n\ :\ A_1x_0\in\mathcal{C}_1 \right\}$;
	\item $x_2=A_1x_1=A_1^2x_0\in\mathcal{C}_2$, i.e., $x_0\in S_{\subalign{&2,0\\&1,2}}=\left\{ x_0\in\R^n\ :\ A_1^2x_0\in\mathcal{C}_2 \right\}$;
	\item $x_3=A_2x_2=A_2A_1^2x_0\in\mathcal{C}_2$, i.e., $x_0\in S_{\subalign{&2,1\\&1,2}}=\left\{ x_0\in\R^n\ :\ A_2A_1^2x_0\in\mathcal{C}_2 \right\}$.
\end{itemize}
Gathering these constraints, we obtain $\displaystyle x_0\in \left( \bigcap_{\tau=0}^1 S_{\subalign{&2,\tau\\&1,2}} \right)\cap \left( \bigcap_{\tau=0}^1 S_{\subalign{&\tau\\&1}} \right)=\Sigma_{\subalign{&2,1\\&1,2}}$.
Reciprocally, if $x_0\in\Sigma_{\subalign{&2,1\\&1,2}}$ the solution initialized from this $x_0$ satisfies:
\begin{itemize}[nosep]
	\item $x_0\in S_{\subalign{&0\\&1}}=\mathcal{C}_1$, thus $x_1=A_1x_0$;
	\item $x_0\in S_{\subalign{&1\\&1}}$, thus $x_1=A_1x_0\in\mathcal{C}_1$ and $x_2=A_1x_1=A_1^2 x_0$;
	\item $x_0\in S_{\subalign{&2,0\\&1,2}}$, thus $x_2=A_1^2x_0\in\mathcal{C}_2$ and $x_3=A_2x_2=A_2A_1^2 x_0$;
	\item $x_0\in S_{\subalign{&2,1\\&1,2}}$, thus $x_3=A_2A_1^2x_0\in\mathcal{C}_2$.
\end{itemize}
We also emphasize  that, without additional information, $x_0\in S_{\subalign{&t\\&i}}$ does not {imply} $x_0\in S_{\subalign{&t-1\\&i}}$.
{In addition, if $x_0\in S_{\subalign{&t_1,\dots,t_{k},\tau\\&i_1,\dots,i_{k},i}}$, without additional conditions (which in fact, leads to $x_0\in \Sigma_{\subalign{&t_1,\dots,t_{k},\tau\\&i_1,\dots,i_{k},i}}$), nothing ensures that the sequence defined~by
	\[x_1=A_{i_1}x_0,\ \ldots,\ x_{t_1}=A_{i_1}^{t_1}x_0,\ x_{t_1+1}=A_{i_2}A_{i_1}x_0,\ \ldots\]
	is solution of~\eqref{eqCLS}.
	In fact, nothing ensures that $x_t$ is in the correct cone (for instance, nothing ensures that $x_0\in\mathcal{C}_{i_1}$).}

We aim to give conditions ensuring that any {nontrivial} solution to \eqref{eqCLS} exhibits no more than $p$ switches.
This is ensured by enforcing that, if the solution associated with an initial condition has already {\textsl{changed} $p$ times of cone}, then it cannot switch anymore.
{That is to say that $\Sigma_{\subalign{&t_1,\dots,t_p,t,0\\&i_1,\dots,i_p,i_{p+1},i}}=\{0\}$ for every $t\ge1$, and every $t_1,\dots,t_p\in\N^*$ and $i_1,\dots,i_{p+1},i\in\{1,\dots,m\}$, with $i_{k+1}\neq i_k$ and $i\neq i_{p+1}$.
	The condition $\Sigma_{\subalign{&t_1,\dots,t_p,t,0\\&i_1,\dots,i_p,i_{p+1},i}}=\{0\}$ for every $t_1,\dots,t_{p+1}\in\N^*$ and every $t\in\N$ explicitly means that the only initial condition for the system~\eqref{eqCLS} such that the corresponding solution of~\eqref{eqCLS} has visited successively the cones $\mathcal{C}_{i_1},\ldots,\mathcal{C}_{i_{p+1}}$, and will visit the cone~$\mathcal{C}_i$, is~$0$.
	We recall that $0 \in \bigcap_{i=1}^m\mathcal{C}_i$, and hence, $\mathcal{I}(0)=\{1,\ldots,m\}^\N$.}
A sufficient condition for this is as follows.
\begin{proposition}\label{res:Simp}
	Given $p\in\N^\star$, $i_1,\dots,i_p,i_{p+1}\in\{1,\dots,m\}$ with $i_{k+1}\neq i_k$ for every $k\in\{1,\dots,p\}$, and given $t_1,\dots,t_p\in\N^\star$, we also set $t_{p+1}=1$.
	If there exist $\kappa\in\{1,\dots,p+1\}$, $1\leq j_1\leq\dots\leq j_\kappa\leq p+1$ and for every $k\in\{1,\dots,\kappa\}$, there exist $r_k\in\N$ and $\tau_{k,1},\dots,\tau_{k,r_k}\in\{0,\dots,t_{j_k}-1\}$, such that $\sum_{k=1}^\kappa r_k\leq n$ and
	\begin{equation}\label{eq:interS}
		\left( \bigcap_{k=1}^\kappa\bigcap_{\ell={1}}^{r_k} S_{\subalign{&t_{j_1},\dots,t_{j_{k-1}},\tau_{k,\ell}\\&i_{j_1},\dots,i_{j_{k-1}},i_{j_k}}}\right)\cap S_{\subalign{&t_1,\dots,t_p,t,0\\&i_1,\dots,i_p,i_{p+1},i}}={\{0\}},
	\end{equation}
	for every $t\in\N^\star$ and every $i\in\{1,\dots,m\}\backslash\{i_{p+1}\}$, then all the solutions to \eqref{eqCLS} initialized in $\Sigma_{\subalign{&t_1,\dots,t_p,0\\&i_1,\dots,i_p,i_{p+1}}}{\setminus\{0\}}$ exhibit no more than~$p$ switches.
\end{proposition}
{We make several comments before proving Proposition \ref{res:Simp}.
	\begin{itemize}[nosep]
		\item As claimed in the paragraph preceding \cref{res:Simp}, the aim of this proposition is to give conditions ensuring that $\Sigma_{\subalign{&t_1,\dots,t_p,t,0\\&i_1,\dots,i_p,i_{p+1},i}}={\{0\}}$ for every $t\in\N^*$.
		\item The proof of this result relies on the fact that given a sequence of sets ${(S_i)}_{i\in I}${, such that $0\in\bigcap_{i\in I}S_i$}, if there exists $J\subset I$ such that $\bigcap_{j\in J}S_J={\{0\}}$ then $\bigcap_{i\in I}S_i={\{0\}}$.
		\item Using the abstract notation of the previous point, we add two constraints on the possible choices of~$J$:
			\begin{enumerate}[nosep]
				\item the cardinality of $J$ is not greater than $n$ (this constraint is $\sum_{k=1}^\kappa r_k\leq n$).
					This constraint is added for the numerical verification of the test (see~\cref{unboundtn}), where a $n\times n$ matrix will be constructed from these sets;
				\item we do not consider sets ${(S_j)}_{j\in J}$ of the form $S_{\subalign{&t_1,\dots,t_p,\tau\\&i_1,\dots,i_p,i_{p+1}}}$ (with $\tau>0$).
					This fact is expressed by the condition $t_{p+1}=1$, together with $\tau_{k,1},\dots,\tau_{k,r_k}\in\{0,\dots,t_{j_k}-1\}$ (that is to say that if $j_{\kappa}=p+1$, we enforce $r_\kappa=1$ and $\tau_{\kappa,r_\kappa}=0$).
			\end{enumerate}
	\end{itemize}
}
\begin{proof}[Proof {of \cref{res:Simp}}]
	For every $t\in\N^*$, we have
    \begin{align*}
	   &\Sigma_{\subalign{&t_1,\dots,t_p,t,0\\&i_1,\dots,i_p,i_{p+1},i}}
		 \subset \Sigma_{\subalign{&t_1,\dots,t_p,0\\&i_1,\dots,i_p,i_{p+1}}}\cap S_{\subalign{&t_1,\dots,t_p,t,0\\&i_1,\dots,i_p,i_{p+1},i}}\\&\hspace{10mm}
		 {= \left(\bigcap_{k=1}^p\bigcap_{\tau=0}^{t_k-1} S_{\subalign{&t_1,\dots,t_{k-1},\tau\\&i_1,\dots,i_{k-1},i_k}}\right)\cap S_{\subalign{&t_1,\dots,t_p,0\\&i_1,\dots,i_p,i_{p+1}}}\cap S_{\subalign{&t_1,\dots,t_p,t,0\\&i_1,\dots,i_p,i_{p+1},i}}}\\&\hspace{10mm}
		 \subset \left( \bigcap_{k=1}^\kappa\bigcap_{\ell={1}}^{r_k} S_{\subalign{&t_{j_1},\dots,t_{j_{k-1}},\tau_{k,\ell}\\&i_{j_1},\dots,i_{j_{k-1}},i_{j_k}}}\right)\cap S_{\subalign{&t_1,\dots,t_p,t,0\\&i_1,\dots,i_p,i_{p+1},i}}.
    \end{align*}
\end{proof}
We are ready to state conditions under which \emph{any} solution to (\ref{eqCLS}), and not only those initialized in $\Sigma_{\subalign{&t_1,\dots,t_p,0\\&i_1,\dots,i_p,i_{p+1}}}$ as in \cref{res:Simp}, exhibits at most $p$ switches.
\begin{theorem}\label{res:SimpComb}
	{Under Assumption~\ref{ass:regular}, g}iven $p\in\N^\star$, if for every $i_1,\dots,i_{p+1}\in\{1,\dots,m\}$ and every $t_2,\dots,t_{p}\in\N^\star$, the conditions of \cref{res:Simp} are satisfied with $t_1=1$, then all the possible {nontrivial} solutions of~\eqref{eqCLS} admit at most $p$ switches.
\end{theorem}
\begin{proof}
	Assume by contradiction that there exists $x_0\in\R^n{\backslash\{0\}}$ such that the solution initialized at $x_0$ switches more that $p$ times.
	Then there exists $i_1,\dots,i_p,i_{p+1},i_{p+2}$ (with $i_{k+1}\neq i_k$) and $t_1,\dots,t_p,t_{p+1}\in\N^\star$ such that $A^{t_1-1}_{i_1}x_0$ belongs to $\Sigma_{\subalign{&1,t_2,\dots,t_{p+1},0\\&i_1,i_2,\dots,i_{p+1},i_{p+2}}}$.
	This leads to a contradiction, since by Proposition~\ref{res:Simp}, this set is {reduced to~$\{0\}$, and by Assumption~\ref{ass:regular}, $A_{i_1}$ is invertible}.
\end{proof}
{\begin{remark}
	Under the assumptions of \cref{res:SimpComb}, we have in fact a stronger result, which is $\max\left\{ \jmath\left( (k_t)_t \right)\ :\ (k_t)_{t\in\N}\in\mathcal{I}(x) \right\}\le p$ for every $x\in\R^n\setminus\{0\}$.
	Obviously, this ensures the claim of \cref{res:SimpComb}: $\varsigma(x)\le p$ for every $x\in\R^n\setminus\{0\}$.
\end{remark}}

From a computational point of view, the main difficulties with the result of Theorem \ref{res:SimpComb} are related to checking the intersection (\ref{eq:interS}) for every $t\in\N^\star$ and every $i\in\{1,\dots,m\}\backslash\{i_{p+1}\}$.
The problem is combinatorial by nature. However, we provide in the sequel, tractable conditions for the case $m=2$.
In \cref{subsect:m=2}, we propose a solution to tackle the fact that $t$ is a priori not bounded.
It consists in using Farkas lemma, which is a particular case of the S-procedure dedicated to linear forms, to transform the problem into assessing the non-negativity of a specific solution to an auxiliary discrete-time system. We propose tractable conditions to check this non-negativity in \cref{sect:positivity-condition}.
As to the combinatorial complexity related to checking the intersection (\ref{eq:interS}) for every $i\in\{1,\dots,m\}\backslash\{i_{p+1}\}$, the situation is manageable as we will manipulate sets defined by linear inequalities, and checking if their intersection is {reduced to~$\{0\}$} can be done using linear programming.

\subsection{When $m=2$}\label{subsect:m=2}

To streamline the discussion and maintain clarity, we focus on the case where $m=2$, i.e., the partition of $\R^{n}$ is made of two cones, while keeping in mind that the results and concepts presented here can be readily extended to the general case when $m\in\N^\star$.
Note that the case where $m=2$ is notoriously difficult \cite{Blondel1999}.

When $m=2$, the cones $\mathcal{C}_1$ and $\mathcal{C}_2$ can be written as in~\eqref{eqcvx12} for some matrix $K\in\R^{1\times {n}}$, which is not the same as in \cref{sect:motivating-example} in general.
{Recall that we assumed that the cones~$\mathcal{C}_i$ are convex, hence, for $m=2$ the two cones~$\mathcal{C}_1$ and~$\mathcal{C}_2$ are half-spaces of~$\R^n$.}

\subsubsection{General result}\label{subsubsect:m=2-general}

Since the indexes $i_k$ used in Section~\mbox{\ref{pg:Helly}}, are such that $i_{k+1}\neq i_k$, the only possibilities, when $m=2$ are the sequences $1,2,1,\dots$ or $2,1,2,\dots$
This observation together with the expression of $\mathcal{C}_1$ and $\mathcal{C}_2$ given in~\eqref{eqcvx12}, leads to set
\begin{multline*}
	\mathcal{S}_{i_1;t_1,\dots,t_p,t}:=\mathcal{S}_{\subalign{&t_1,\dots,t_p,t\\&[i_1],\dots,[i_1+p]}}\\
	=\left\{ x\in\R^n\ :\ (-1)^{i_1+p}KA_{[i_1+p+1]}^tA_{[i_1+p]}^{t_p} \ldots
	A_{[i_1]}^{t_1}x\geq 0 \right\},
\end{multline*}
where, for every $i\in\N$, we have defined $[2i+1]=1$ and $[2i]=2$.
Similarly, we define $\Sigma_{i_1;t_1,\dots,t_p,t}$.
Given $i_1\in\{1,2\}$ if we start from an initial condition in $\mathcal{C}_{i_1}$, the solution stays for the remaining times in $\mathcal{C}_{i_1}$ or enters in $\mathcal{C}_{[i_1+1]}$.
Hence, the maximal number of switches starting from $\mathcal{C}_{i_1}$ is bounded by one plus the maximal number of switches starting from~$\mathcal{C}_{[i_1+1]}$.
This observation leads to the next corollary of \cref{res:SimpComb}.
\begin{corollary}\label{cor:m-2}
	Under \cref{ass:regular}, consider the system (\ref{eqCLS}) with $m=2$.
	Given $p\in\N^\star$, if there exists $i_1\in\{1,2\}$ such that for every $t_2,\dots,t_p\in\N^\star$, the conditions of \cref{res:Simp} are satisfied with $t_1=1$, then each {nontrivial} solution of~\eqref{eqCLS} admits at most $p+1$ switches.
\end{corollary}

\subsubsection{Tractable conditions}\label{pg:Farkas}

To make the paper self-contained, we recall Farkas lemma in the form in which we will use it, although other formulations are possible, see e.g.~\cite{LaraHir} and \cite[Corollary 4.3]{Bertsim97} for more detail.
\begin{lemma}[Farkas lemma]\label{res:Farkas}
	Let $\ell,n\in\N^\star$ and $M_0,\ldots,M_{\ell}\in\R^n$.
	We have,
	\begin{multline*}
		\left\{ x\in\R^n\ \,:\,\ x^{\top} M_0\ge0 \right\}\subset\bigcap_{i=1}^\ell \left\{ x\in\R^n\ \,:\,\ x^{\top} M_i\ge0 \right\}\\[-1mm]
		\quad\Leftrightarrow\quad
		\exists { \alpha_1,\dots,\alpha_\ell\in\R_+\ \text{s.t.}\ M_0=\sum_{i=1}^\ell\alpha_i M_i.}
	\end{multline*}
\end{lemma}
In other words, Lemma \ref{res:Farkas} states that, given any $x\in \R^n$,
$$x^{\top} M_i\ge 0 \quad \forall i\in\{1,\dots,\ell\} \quad \Longrightarrow \quad x^{\top} M_0 \ge 0$$
if and only if there exist non-negative numbers {$\alpha _1,\alpha _2, \ldots,\alpha_\ell$} such that {$M_0 = \sum_{i=1}^{\ell} \alpha _i M_i$.}
In the sequel, to verify condition~(\ref{eq:interS}), we consider without loss of generality that \mbox{$\sum_{k=1}^\kappa r_k=n$}.
This is justified because, if (\ref{eq:interS}) holds with $\sum_{k=1}^\kappa r_k = \nu$ with $\nu<n$, it also holds with $\sum_{k=1}^\kappa r_k = n$.
This can be achieved by including any $n-\nu$ sets from the available sets{, we recall that $\sum_{k=1}^\kappa r_k$ is the number of {sets} used for testing the emptiness condition~\eqref{eq:interS}}. 
For $t_1,\dots,t_{n-1}\in\N^\star$, we introduce the following matrices, which all belong to~$\R^{1\times n}$,
\begin{equation}\label{eq:St}
	\begin{array}{lr}
		N_{i_1;t}
		=(-1)^{i_1+1}KA_{[i_1]}^t, 
		& \text{for }0\le t<t_1,\\[1mm]
		N_{i_1;t_1,t}
		=(-1)^{i_1+2}KA_{[i_1+1]}^tA_{[i_1]}^{t_1}, 
		& \text{for }0\le t<t_2,\\[1mm]
		N_{i_1;t_1,t_2,t}
		=(-1)^{i_1+3}KA_{[i_1+2]}^tA_{[i_1+1]}^{t_2}A_{[i_1]}^{t_1},\hskip -2mm & \text{for }0\le t<t_3, \\
		\hskip 5mm ~\vdots &\\
		N_{i_1;t_1,\dots,t_{n-2},t}
		=(-1)^{i_1+n-1}KA_{[i_1+n-2]}^tA_{[i_1+n-3]}^{t_{n-2}}\dots A_{[i_1]}^{t_1},\hskip -10 cm\\
		& \text{for }0\le t<t_{n-1}, \\
		N_{i_1;t_1,\dots,t_{n-2},t_{n-1},0}
		=(-1)^{i_1+n}KA_{[i_1+n-2]}^{t_{n-1}}\dots A_{[i_1]}^{t_1}.\hskip -10 cm
	\end{array}
\end{equation}
Observe that these matrices verify
\begin{equation}\label{eq:St2}
	S_{i_1;t_1,\dots,t_p,t}=\left\{ x\in\R^n\ \,:\,\ N_{i_1;t_1,\dots,t_p,t}x\ge0 \right\}.
\end{equation}

To verify (\ref{eq:interS}) we investigate the intersection of $S_{i_1;1,t_2,\dots,t_{n-1},t,0}$ with $n$-sets taken from~\eqref{eq:St}.
We notice that in~\eqref{eq:St}, there are exactly $1+\sum_{i=1}^{n-1} t_i$ row vectors defined  (recall that $t_1,\ldots,t_{n-1}\in\N^*$, hence, $1+\sum_{i=1}^{n-1} t_i\ge n$).\\
{In addition, we also have (with $t_1=1$)}
$$S_{i_1;1,t_2,\dots,t_{n-1},t,0}=\left\{ x\in\R^n\ \,:\,\ N_{i_1;{1},t_2,\dots,t_{n-1},t,0}x\ge0 \right\},$$ with
\begin{multline*}
	N_{i_1;{1, t_2}\dots,t_{n-1},t,0}=(-1)^{i_1+n}KA_{[i_1+n-1]}^tA_{[i_1+n-2]}^{t_{n-1}}\dots {A_{[i_1]}}.
\end{multline*}

The next result gathers \cref{cor:m-2,res:Farkas}.
\begin{theorem}\label{unboundtn}
	{Under \cref{ass:regular},} suppose there exists $i_1\in\{1,2\}$ such that for every $t_2,\dots,t_{n-1}\in\N^\star$, we either have $\Sigma_{i_1;1,t_2,\dots,t_{n-1},0}={\{0\}}$ or {one is able to build } an invertible matrix $\mathcal{N}\in\R^{n\times n}$ whose lines are taken from~\eqref{eq:St} such that 
	\begin{enumerate}[label=(\roman*)]
		\item\label{it:unboundtn:1} $\beta _0 = -(KA_{i_{[i_1+n-1]}}{\cal M}\mathcal{N}^{-1})^{\top} $ is non-negative,
		\item\label{it:unboundtn:2} the solution to
			\begin{equation}
				\beta _{t+1} = {\cal L} \beta_t,
				\quad \mbox{with}\ \
				{\cal L}^{\top} = \mathcal{N}{\cal M}^{-1} A_{i_{[i_1+n-1]}} {\cal M}\mathcal{N}^{-1}
				\label{eqbeta}
			\end{equation}
			initialized with $\beta _0 = -(KA_{i_{[i_1+n-1]}}{\cal M}\mathcal{N}^{-1})^{\top} \ge 0$ remains non-negative for all future time,
	\end{enumerate}
    where we have set,
    $$\mathcal{M}\defeq (-1)^{i_1+n}A_{[i_1+n-2]}^{t_{n-1}}\dots A_{[i_1+1]}^{t_2}A_{i_1}.$$
	Then any solution to (\ref{eqCLS}) exhibits no more than $n$ switches.
\end{theorem}
\begin{proof} 
	Using \cref{cor:m-2}, any {nontrivial} solution to (\ref{eqCLS}) exhibits no more than $n$ switches if \eqref{eq:interS} is satisfied with $t_1=1$ and \mbox{$p=n-1$}.
	This means checking that the intersection of $ S_{\subalign{&1,\dots,t_{n-1},t,0\\&i_1,\dots,i_{n-1},i_{n},i}}$ with the sets characterized by~\eqref{eq:St} and~\eqref{eq:St2} is {reduced to $\{0\}$}.
	To check this condition, it is enough to exhibit $n$ sets such that, for every $t\in\N^\star$, their intersection with $S_{i_1;1,t_2,\dots,t_{n-1},t,0}$ is {reduced to $\{0\}$}.
    This is equivalent to find $n$ rows in~\eqref{eq:St} such that the matrix $\mathcal{N}\in\R^{n\times{n}}$ formed by these $n$ rows is such that
	\begin{equation}\label{eqNj}
		\mathcal{N}x\ge0
	\end{equation}
	implies
	\begin{equation*}
		(-1)^{i_1+n}KA_{[i_1+n-1]}^tA_{[i_1+n-2]}^{t_{n-1}}\dots A_{[i_1]}x\le0,
		\qquad\forall t \in \N^\star,
	\end{equation*}
	i.e., implies
	\begin{equation}\label{eqCns2b}
		-KA_{[i_1+n-1]}^t\mathcal{M}x\ge0,
		\qquad\forall t \in \N^\star.
	\end{equation}
	Invoking Lemma \ref{res:Farkas}, checking that (\ref{eqNj}) implies (\ref{eqCns2b}) reduces to check whether there exists $\beta_t\in\R_+^n$ such that
	$$-KA_{[i_1+n-1]}^t\mathcal{M}=\beta_t^{\top} \mathcal{N}, \qquad\forall t \in \N^\star,$$
	that is to say, with $\mathcal{N}$ invertible,
	\begin{equation}\label{eqCns2pos}
		\beta_t^{\top} =-KA_{[i_1+n-1]}^t\mathcal{M}\mathcal{N}^{-1}\ge0, \qquad\forall t \in \N^\star.
	\end{equation}
	To conclude the proof, we show that (\ref{eqCns2pos}) is equivalent to check the positivity of a solution of a discrete time dynamics. 
	Indeed, as the invertibility of $A_1$ and $A_2$ implies that $\mathcal{M}$ is invertible, we have
	\begin{align*}
		\beta_{t+1}^{\top} &= -KA_{[i_1+n-1]}^{t+1}\mathcal{M}\mathcal{N}^{-1}\\
		&=-KA_{[i_1+n-1]}^t\mathcal{M}\mathcal{N}^{-1}\mathcal{N}{\mathcal{M}^{-1}}A_{[i_1+n-1]}\mathcal{M}\mathcal{N}^{-1}\\
		&=\beta_t^{\top} \mathcal{N}\mathcal{M}^{-1}A_{[i_1+n-1]}\mathcal{M}\mathcal{N}^{-1}\\
		&=(\mathcal{L}\beta_t)^{\top}.
	\end{align*}
\end{proof}

This theorem states that to verify the non-negativity of a single solution to the auxiliary linear discrete-time system described in (\ref{eqbeta}) is enough to ensure that any solution to~(\ref{eqCLS}) switches no more than $n$ times. 
This property  is different from the classical non-negativity problem, where any solution to a system must remain non-negative (component-wise) for any non-negative initial condition. In this classical setting, it is well known that, for a discrete-time system to be positive, a necessary and sufficient condition is that the entries of the dynamical matrix have to be positive \cite{Farina2000}. This does not correspond to our problem, and we will see in Section~\ref{sect:applications} that the entries of the dynamical matrix ${\cal L}$ are not necessarily non-negative, so that we cannot invoke the results from~\cite{Farina2000} in general. This means new results are needed to ensure the non-negativity of  system (\ref{eqbeta}) for a given initial condition: this is addressed in the next section. 

\section{Non-negativity analysis}\label{sect:positivity-condition}

We derive in this section conditions under which the solution to (\ref{eqbeta}) initialized with a given $\beta_0\ge0$ is non-negative.
The results are based on a function analysis and allows deriving numerically tractable conditions.

\subsection{Non-negativity theorem}\label{subsect:non-negativity-thm}

The next theorem provides easy-to-test conditions to ensure the satisfaction of item (ii) of Theorem \ref{unboundtn}.
The result of \cref{TestJerome} is based on the necessary condition given in \cref{res:TestSimp}, which is given below.
\begin{theorem} \label{TestJerome}
	Consider system (\ref{eqbeta}) with a given $\beta_0 \geq~0$ and assume that $\sigma(A_{[i_1+n-1]})=\{\lambda_1,\dots,\lambda_n\}$, with $\lambda_1>\dots>\lambda_n>0$.
	Then $\mathcal{L}$ is diagonalizable, and we set $v_1,\dots,v_n$ the eigenvectors of~$\mathcal{L}$, and $\gamma_1,\dots,\gamma_n\in\R$ such that $\beta_0=\sum_{i=1}^n\gamma_iv_i$.
	Let us finally set $\rho_i = \gamma_i v_i\in\R^n$ and $\mu_i=-\ln(\lambda_i/\lambda_1)$ with $i=1, \ldots, n$ and assume that $\rho_1 \geqslant 0$.
	Then, the solution $\beta _t$ is non-negative for any $t \in \N$ if for every $\ell\in\{1,\dots,n\}$, one of the conditions of \cref{res:TestSimp} is satisfied, with $z_i$ the $\ell$-th component of $\rho_i$.
\end{theorem}
\begin{proof}
	According to the definition of ${\cal L}$ in (\ref{eqbeta}), the eigenvalues of ${\cal L}$ are also eigenvalues of $A_{[i_1+n-1]}$.
	They are distinct, and the $n$ eigenvectors $v_i$ are independent.
	The solution of the linear discrete-time system (\ref{eqbeta}) is a weighted linear combination of modal solutions:
	$$\beta _t = \sum_{i=1}^{n} \gamma_i v_i \lambda _i^t.$$
	
	To conclude the proof of this theorem, we have to find {conditions} on $( \rho_1,\dots, \rho_n)$ such that
	\begin{equation}\label{eq:Obj0z}
		\beta _t = \sum_{i=1}^{n} \gamma_i v_i \lambda _i^t =	\lambda_1^t\rho_1+\dots+\lambda_n^t\rho_n\geqslant0\qquad (t\in\N).
	\end{equation}
	We recall that $\rho_i\in\R^n$ and inequality \eqref{eq:Obj0z} has to be understood component-wise, i.e., $[\beta_t]_\ell\ge0$ for every $\ell\in\{1,\dots,n\}$.
	According to the definition of $\mu _i$, for every $i\in\{2, \ldots ,n\}$, we have $0<\mu_2<\ldots <\mu_n$.
	Hence, the problem is equivalent with finding a condition on $(\rho_1,\dots,\rho_n)\in\R^n$ such that
	\begin{equation}\label{eq:Objn1}
		\rho_1\ge -\sum_{i=2}^n e^{-\mu_it} \rho_i\qquad (t\in\N).
	\end{equation}
	Note that \eqref{eq:Objn1} holds if we have
	\begin{equation*}
		\rho_1\ge -\sum_{i=2}^n e^{-\mu_is} \rho_i\qquad (s\in\R_+).
	\end{equation*}
	We conclude the proof using \cref{res:TestSimp}
\end{proof}

\begin{proposition}\label{res:TestSimp}
	Let $n\in\N^*$, $0<\mu_2<\dots<\mu_n$ and \mbox{$z_1,\dots,z_n\in\R$}.
	{If $z_1\ge0$, $z_1+\dots+z_n\ge0$ and one of the following conditions holds:}
	\begin{enumerate}[label={(\roman*)},nosep]
		\item\label{it:g1} $z_2,\dots,z_n\ge0$;
		\item\label{it:g2} $z_2,\dots,z_n\le0$;
		\item\label{it:g3} there exists $j\in\{2,\dots,n\}$ such that 
			\begin{equation}\label{eq:casn}
				z_1\ge-\sum_{\substack{i=2\\i\neq j}}^n\left( 1-\frac{\mu_i}{\mu_j} \right)z_ie^{-\mu_i s}, \quad \forall s\in\R_+.
			\end{equation}
	\end{enumerate}
	{Then, the function $\varphi(s)=z_1+z_2e^{-\mu_2 s}+\dots+z_ne^{-\mu_ns}$ defined for every $s\in\R_+$ is non-negative on $\R_+$.}
\end{proposition}
\begin{proof}
	Observe that $\varphi(0)=z_1+\dots+z_n$ and $\lim_{s\to\infty}\varphi(s)=z_1$.
	They are both non-negative in the setup of the statement.
	Obviously, if \ref{it:g1} is satisfied, $\varphi\ge0$ on~$\R$.
	If~\ref{it:g2} is satisfied, we have $\varphi'\ge0$ on $\R$ and $\varphi(0)=z_1+\dots+z_n\ge0$.\\
	Without the sign constraints on $z_2,\dots,z_n$, we have either $\inf_{\R_+}\varphi=\varphi(0)=z_1+\dots+z_n$ or $\inf_{\R_+}\varphi=\lim_{s\to\infty}\varphi(s)=z_1$ or $\inf_{\R_+}\varphi=\varphi(s_0)$ for some $s_0\in\R_+^*$ such that $\varphi'(s_0)=0$.
	Hence, to guarantee the nonnegativity of $\varphi$ on $\R_+$ it remains to check the positivity of $\varphi(s_0)$ in the last case.
	Observe that, given $j\in\{2,\dots,n\}$, $\varphi'(s_0)=0$ is equivalent to $\mu_jz_je^{-\mu_j s_0}=\sum_{\substack{i=1\\i\neq j}}^n -\mu_iz_ie^{-\mu_i s_0}$.
	Hence, we have $\varphi(s_0)=z_1+\sum_{\substack{i=2\\i\neq j}}\left( 1-\frac{\mu_i}{\mu_j} \right)z_ie^{-\mu_i s_0}$, and \ref{it:g3} ensures that $\varphi(s_0)\ge0$.
\end{proof}
\cref{TestJerome} allows deriving numerically tractable conditions by iteratively applying (\ref{eq:casn}) to any chosen value of $n$. As a title of example, consider the case $n=4$.
By iteratively applying~(\ref{eq:casn}), one can conclude that $\beta _t \geq 0$ for all $ t$ if
\begin{equation}\label{corro4}
	\begin{array}{c}
		\rho_1\ge0, \quad 
		\rho_1+\left( 1-\frac{\mu_3}{\mu_2} \right)\rho_3+\left( 1-\frac{\mu_4}{\mu_2} \right)\rho_4\ge0, \\
		\mbox{and }	\rho_1+\left( 1-\frac{\mu_4}{\mu_3} \right)\left( 1-\frac{\mu_4}{\mu_2} \right)\rho_4\ge0.
	\end{array}
\end{equation}
\smallskip
To check that this is true, first iterate \eqref{eq:casn} with $n=4$.
As $\left( 1-\frac{\mu_4}{\mu_3} \right)\left( 1-\frac{\mu_4}{\mu_2} \right)\rho_4\ge0$ or $\left( 1-\frac{\mu_4}{\mu_3} \right)\left( 1-\frac{\mu_4}{\mu_2} \right)\rho_4\le0$, this together with $\rho_1\ge0$ and $\rho_1+\left( 1-\frac{\mu_4}{\mu_3} \right)\left( 1-\frac{\mu_4}{\mu_2} \right)\rho_4\ge0$ ensures
$$\rho_1\ge -\left( 1-\frac{\mu_4}{\mu_3} \right)\left( 1-\frac{\mu_4}{\mu_2} \right)\rho_4e^{-\mu_4 s}\qquad(s\in\R_+).$$
This, together with $\rho_1\ge0$ and $\rho_1+\left( 1-\frac{\mu_3}{\mu_2} \right)\rho_3+\left( 1-\frac{\mu_4}{\mu_2} \right)\rho_4\ge0$ ensures
$$\rho_1\ge-\left( 1-\frac{\mu_3}{\mu_2} \right)\rho_3e^{-\mu_3 s}-\left( 1-\frac{\mu_4}{\mu_2} \right)\rho_4e^{-\mu_4 s}\qquad (s\in\R_+).$$
Finally, this last fact, together with $\rho_1\ge0$ and $\beta _0 = \rho_1+\rho_2+\rho_3+\rho_4\ge0$ ensures
$$\rho_1\ge -\rho_2e^{-\mu_2 s}-\rho_3e^{-\mu_3 s}-\rho_4e^{-\mu_4 s}\qquad(s\in\R_+).$$
We will exploit \eqref{corro4} for the insulin infusion example in \cref{subsect:applications-insulin}.

\subsection{Tailored results for small values of $n$}

For small values of $n$, the function analysis in the proof of \cref{res:TestSimp} can be tailored to derive tighter conditions.
{%
In this paragraph, we only present the cases $n\in\{2,3,4\}$.
But similar results could have been obtained for larget values of~$n$.
}

When $n=2$, it is easy to see that $\varphi(s)=z_1+z_2e^{-\mu_2 s}\ge0$ for every $s\in\R_+$ (with $\mu_2\ge0$) if and only if $z_1\ge0$ and $z_1+z_2\ge0$.
The next propositions provide tailored results for $n=3$ and $n=4$.
\begin{proposition}[Case $n=3$]\label{corron3}
	With the notation of \cref{res:TestSimp} with $n=3$, we have $\varphi(s)\ge0$ for all $s\in\R_+$ if and only if $z_1\ge0$, $z_1+z_2+z_3\ge0$ and one of the following conditions is satisfied:
	\begin{enumerate}[label={(\roman*)},nosep]
		\item $z_2\ge0$.
		\item $z_2<0$, $z_3\le0$.
		\item $z_2<0$, $z_3>0$, $\mu_2z_2+\mu_3z_3\le0$.
		\item $z_2<0$, $z_3>0$, $\mu_2z_2+\mu_3z_3>0$ and
			$$z_1+z_2\left( \frac{-\mu_2z_2}{\mu_3z_3} \right)^{\frac{\mu_2}{\mu_3-\mu_2}}+z_3\left( \frac{-\mu_2z_2}{\mu_3z_3} \right)^{\frac{\mu_3}{\mu_3-\mu_2}}\ge0.$$
	\end{enumerate}
\end{proposition}
\begin{proof}
	Note that the cases where $z_2$ and $z_3$ have the same signs are treated in items~\ref{it:g1} and~\ref{it:g2} of \cref{res:TestSimp}.
	Hence, it remains to prove the result when $z_2$ and $z_3$ do not have the same sign.
	Recall that $\varphi'(s)=-\mu_2z_2e^{-\mu_2s}-\mu_3z_3e^{-\mu_3s}$, and $\sign \varphi'=\sign g$, with $g=-z_2-\frac{\mu_3}{\mu_2}z_3e^{-(\mu_3-\mu_2)s}$.
	We also have $g'(s)=\frac{\mu_3}{\mu_2}z_3(\mu_3-\mu_2)e^{-(\mu_3-\mu_2)}$, thus $\sign g'=\sign z_3$.
	When~$z_2$ and~$z_3$ do not have the same sign, there exists one and only one $s_0\in\R$ such that $g(s_0)=0$ (see \cref{fig:oop}).
	We thus have the following situations.
	\begin{itemize}
		\item If $z_2>0$ and $z_3<0$, we have the situation described by \cref{fig:pm}.
			We observe that $\varphi(s)\ge\min\{\varphi(0),z_1\}\ge0$ for every $s\in\R_+$.
		\item If $z_2<0$ and $z_3>0$, we have the situation described by \cref{fig:mp}.
			We observe that $\inf_{\R_+}\varphi=\begin{cases}
				\varphi(0) & \text{if } s_0\le0,\\
				\varphi(s_0) & \text{otherwise.}
			\end{cases}$
			But $s_0$ is such that $g(s_0)=0$, i.e., $s_0=\frac{-1}{\mu_3-\mu_2}\ln\frac{-\mu_2z_2}{\mu_3z_3}$.
			We have $s_0\le0$ if and only if $\mu_2z_2+\mu_3z_3\le0$, and $\varphi(s_0)=z_1+z_2\left( \frac{-\mu_2z_2}{\mu_3z_3} \right)^{\frac{\mu_2}{\mu_3-\mu_2}}+z_3\left( \frac{-\mu_2z_2}{\mu_3z_3} \right)^{\frac{\mu_3}{\mu_3-\mu_2}}$.
	\end{itemize}
	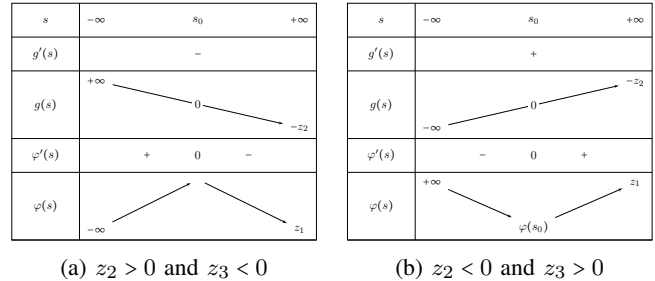
\begin{figure}[!htb]
		\centering
			\begin{subfigure}{0.49\columnwidth}
			\centering
			\scalebox{0.45}{
				\begin{tikzpicture}
					\tkzTabInit{$s$/1,$g'(s)$/1,$g(s)$/2,$\varphi'(s)$/1,$\varphi(s)$/2}{$-\infty$,$s_0$,$+\infty$}
					\tkzTabLine{,,-,,}
					\tkzTabVar{+/$+\infty$,R/,-/$-z_2$}
					\tkzTabIma{1}{3}{2}{$0$}
					\tkzTabLine{,+,0,-,}
					\tkzTabVar{-/$-\infty$,+/,-/$z_1$}
				\end{tikzpicture}
			}
			\caption{$z_2>0$ and $z_3<0$}
			\label{fig:pm}
		\end{subfigure}
		\begin{subfigure}{0.49\columnwidth}
			\centering
			\scalebox{0.45}{
				\begin{tikzpicture}
					\tkzTabInit{$s$/1,$g'(s)$/1,$g(s)$/2,$\varphi'(s)$/1,$\varphi(s)$/2}{$-\infty$,$s_0$,$+\infty$}
					\tkzTabLine{,,+,,}
					\tkzTabVar{-/$-\infty$,R/,+/$-z_2$}
					\tkzTabIma{1}{3}{2}{$0$}
					\tkzTabLine{,-,0,+,}
					\tkzTabVar{+/$+\infty$,-/$\varphi(s_0)$,+/$z_1$}
				\end{tikzpicture}
			}
			\caption{$z_2<0$ and $z_3>0$}
			\label{fig:mp}
		\end{subfigure}
		\caption{Different situations in the proof of \cref{corron3}.}
		\label{fig:oop}
	\end{figure}
\end{proof}

\begin{proposition}[Case $n=4$]\label{corro4t}
	With the notations used in \cref{res:TestSimp} with $n=3$, we have $\varphi(s)\ge0$ for all $s\in\R_+$ if $z_1\ge0$, $z_1+z_2+z_3+z_4\ge0$ and one of the following conditions is satisfied:
	\begin{enumerate}[label={(\roman*)},nosep]
		\item\label{it:1} $z_4\le0$, $z_3\le0$.
		\item\label{it:2} $z_4\le0$, $z_3>0$, $z_2\ge0$.
		\item\label{it:3} $z_4\le0$, $z_3>0$, $z_2<0$ and\\
            $\frac{\mu_4}{\mu_2}\left( \frac{\mu_4-\mu_2}{\mu_3-\mu_2}-1\right)e^{\frac{\mu_4-\mu_2}{\mu_4-\mu_3}\ln\left(\frac{z_3}{|z_4|}\frac{\mu_3}{\mu_4}\frac{\mu_3-\mu_2}{\mu_4-\mu_2}\right)}z_4\ge z_2$.
		\item\label{it:5} $z_4>0$, $z_3\ge0$, $z_2\ge0$.
		\item\label{it:7} $z_4>0$, $z_3<0$, $z_2\le0$ and $|z_3|\mu_3(\mu_3-\mu_2)\ge z_4\mu_4(\mu_4-\mu_2)$.
		\item\label{it:8} $z_4>0$, $z_3<0$, $z_2>0$ and\\
            $\frac{\mu_4}{\mu_2}\left( \frac{\mu_4-\mu_2}{\mu_3-\mu_2}-1\right)e^{\frac{\mu_4-\mu_2}{\mu_4-\mu_3}\ln\left(\frac{|z_3|}{z_4}\frac{\mu_3}{\mu_4}\frac{\mu_3-\mu_2}{\mu_4-\mu_2}\right)}z_4\le z_2$. \hfill $\Box$
	\end{enumerate}
\end{proposition}
The proof of this proposition follows the lines of the one of \cref{corron3}.
For the sake of brevity, we do not detail it here.

\section{Applications}\label{sect:applications}

Before applying the results to the insulin infusion example in Section \ref{subsect:applications-insulin}, we provide in Section \ref{subsect:ex-second-order} a second order example to illustrate the fact that the application of \cref{cor:m-2} and \cref{unboundtn} is easy in $\R ^2$. It reduces to check the invertibility of a second order matrix and the non-negativity of a second order dynamics. 

\subsection{Second order example}\label{subsect:ex-second-order}
Consider system (\ref{eqCLS}) with $n=2$, $m=2$,
\begin{equation}
	\begin{array}{l}
		A_1 = \left[ \begin{array}{ cc}
				13 & 12.5 \\
		-12.5 & -12 \end{array} \right], \quad 
		A_2 = \left[ \begin{array}{cc}
				0.93 & 0 \\
				0    & 0.95
		\end{array} \right]\\
		K=\left[ \begin{array}{ cc} -8.7 & -9.2 \end{array} \right],
	\end{array}
\end{equation}
and $\mathcal{C}_1$ and $\mathcal{C}_2$ as in (\ref{eqcvx12}). 
The solutions associated with various initial conditions selected from the unit circle are represented in Figure~\ref{f1second} as well as the corresponding number of switches they exhibit. These solutions predominantly reside within ${\cal C}_2$ and only sporadically venture into ${\cal C}_1$. We observe on this example that $\R^2$ is partitioned in three cones, which are respectively defined by the set of initial conditions that result in exactly $0$, $1$ and $2$ switches.
These cones are formed by the line $Kx=0$ and the half-line $\R_+\left[ 1 \,\, 0 \right]^\top$.
Observe also that $\left[1 \,\, 0\right]^\top\in\mathcal{C}_2$ is an eigenvector of~$A_2$. {Moreover, all points initialized in $\mathcal{C}_1$ reach the invariant region of $\mathcal{C}_2$ after a single iteration as highlighted in Figure~\ref{f1secondsingle} for $x_0=[0~-1]^\top$.}\\

\begin{figure}[ht]
	\begin{center}
        \includegraphics[width=\the\columnwidth]{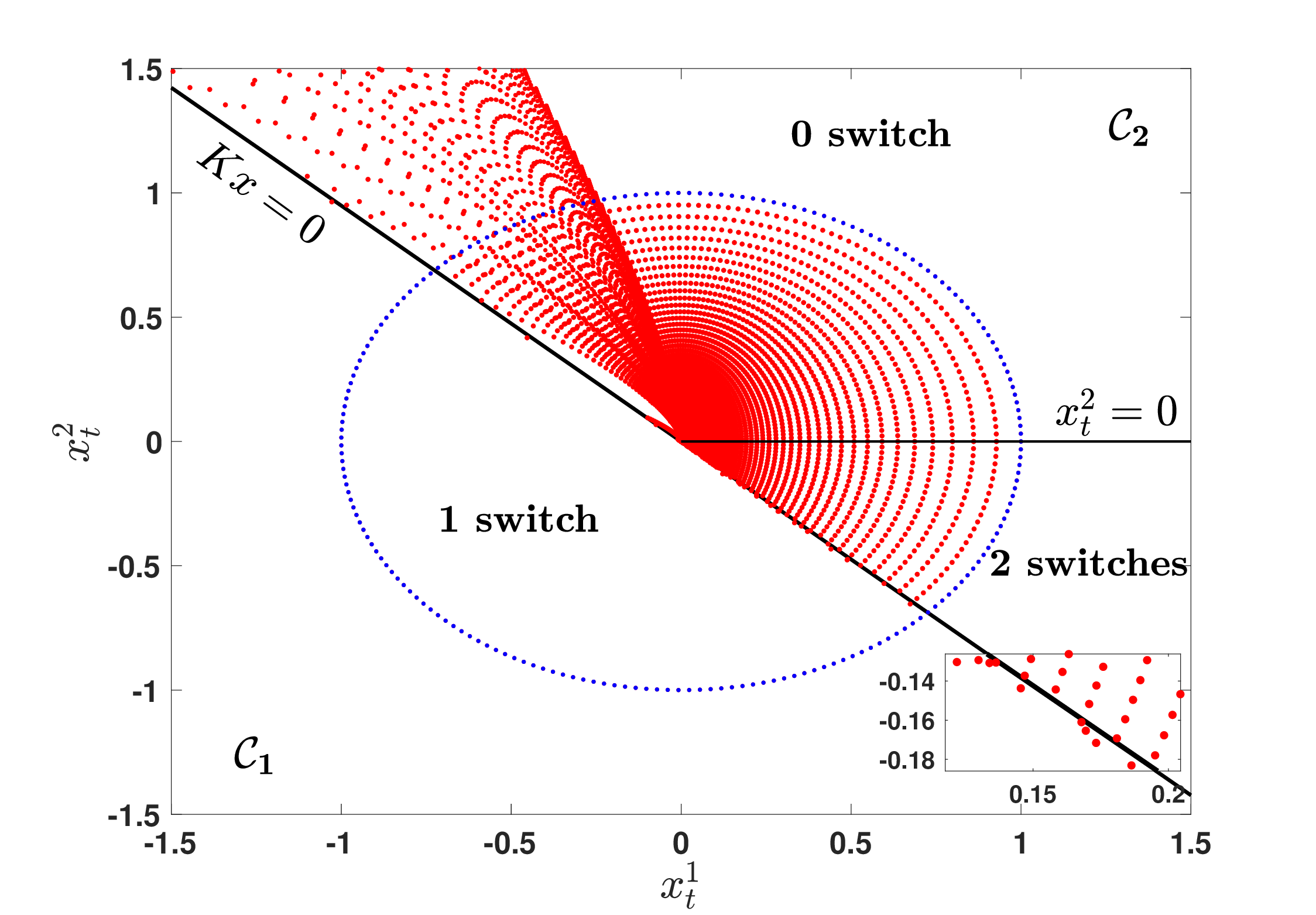}
    \end{center}
		\caption{\label{f1second}
            {Phase portrait in red with $x_0$ on the unit circle.}
        }
\end{figure}

\begin{figure}[ht]
	\begin{center}
        \includegraphics[width=\the\columnwidth]{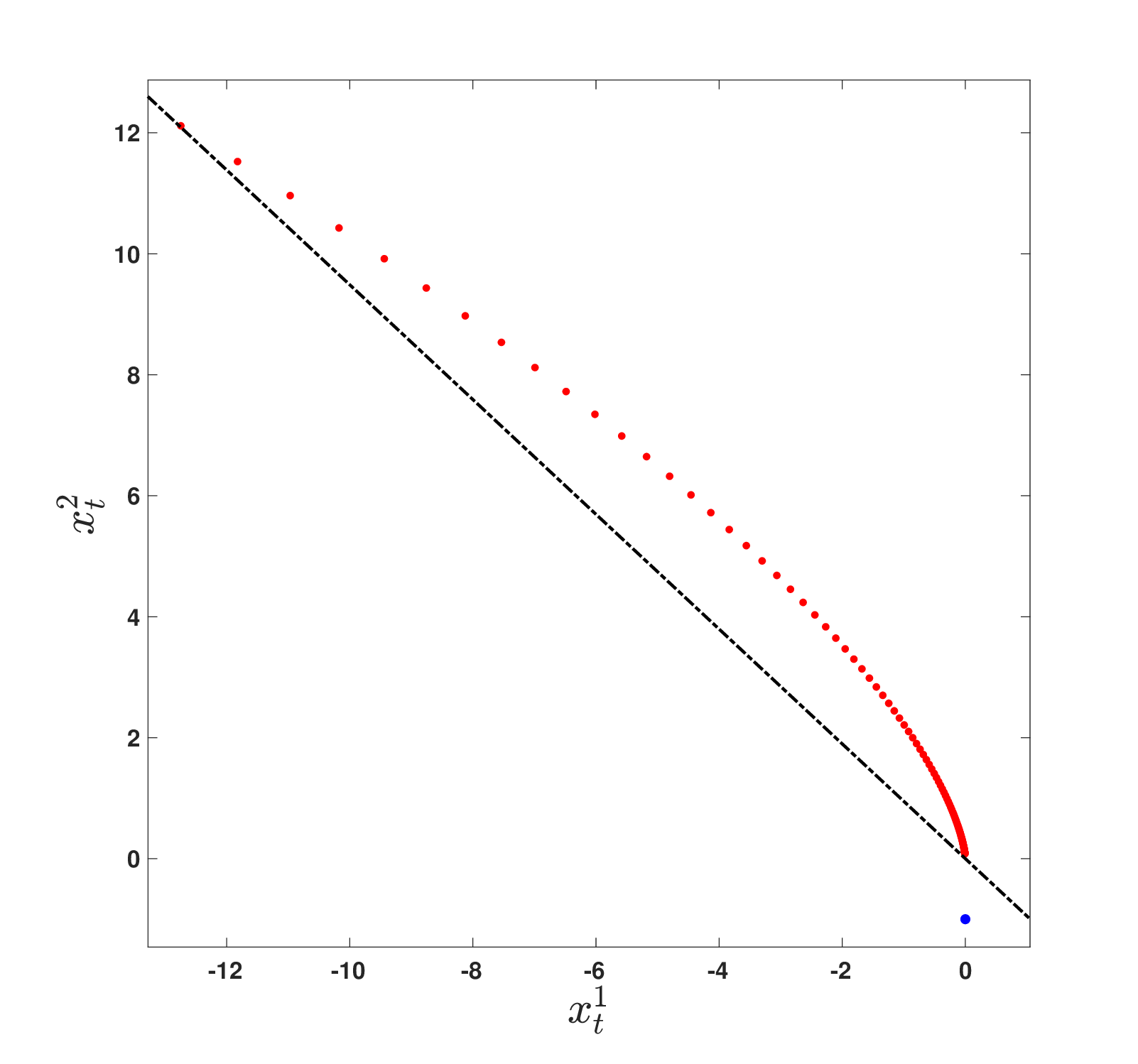}
    \end{center}
		\caption{\label{f1secondsingle}
            {{Phase portrait in red with $x_0=[0, -1]^\top$ ($x_0$ in ${\cal C}_1$).}}
        }
\end{figure}

Let us illustrate the theoretical results of Section \ref{subsect:m=2}. The application of Theorem~\ref{unboundtn} reduces to check that the following implication holds true for $x\in\R^{n}$ 
$$
Kx \geq 0 \mbox{ and } KA_1x \leq 0 \quad \Longrightarrow \quad \forall t\in\N^*, \quad KA_2^{t}A_1x \leq 0
$$
For $t=1$, using Lemma \ref{res:Farkas}, this is equivalent to the existence of $\alpha_1^1, \alpha_2^1 \geq 0 $ such that
$$
-KA_2A_1 = \alpha_1^1 K - \alpha_2^1 KA_1 .
$$
The pair ($-A_1, K$) is observable. Hence,
$$
\begin{array}{ r@{\>}c@{\>}l}
	\left[ \begin{array}{cc}  \alpha_1^1 & \alpha_2^1 \end{array} \right] 
	&=& -KA_2A_1 \left[ \begin{array}{cc} K \\ -KA_1 \end{array} \right] ^{-1}\\
	&=& \left[ \begin{array}{cc} 0.1281 & 1.5541 \end{array} \right].
\end{array}
$$
Now, we have to study the system
$$
\beta _{t+1} = {\cal L} \beta _{t}, \quad \beta _0 = \left[ \begin{array}{cc} \alpha_1^1 & \alpha_2^1 \end{array} \right]^{\top} \geq 0
$$
$$
\mbox{with } {\cal L} = \left[ \begin{array}{rc}
		0.3259  & 0.1281 \\
		-2.9443 & 1.5541
 \end{array} \right], \quad \beta _0 = \left[ \begin{array}{c} 0.1281 \\ 1.5541 \end{array} \right]. 
$$
The eigenvalues of ${\cal L}$, which are also eigenvalues of $A_2$, and the associated eigenvectors are given by
$\lambda_1 = 0.95$, 
$\lambda_2 = 0.93$,
$v_1 = \left[ \begin{array}{cc} -0.2010 & -0.9796 \end{array} \right]^\top$.
$v_2 = \left[ \begin{array}{cc} -0.2074 & -0.9783 \end{array} \right]^\top$,
We have $\beta _t = \sum_{i=1}^{n} \gamma_i v_i \lambda _i^t$, with $\beta_0 = \sum_{i=1}^{n} \gamma_i v_i \geq 0$,
where $\gamma_1$ and $\gamma_2$ are given by
$$
\left[ \begin{array}{c} \gamma_1 \\ \gamma_2 \end{array} \right]
= \left[ \begin{array}{ccc}  v_1 & v_2 \end{array} \right] ^{-1} \beta _0
= \left[ \begin{array}{r}  28.7099 \\  -30.2574  \end{array} \right].
$$
$$
\mbox{and } \gamma_1 v_1 =	\left[ \begin{array}{r}  6.0818 \\29.6399 \end{array} \right],
\quad
\gamma_2 v_2 =	\left[ \begin{array}{r} -5.9537 \\-28.0858 \end{array} \right],
$$
with $\lambda_1 > \lambda_2 > 0$, we conclude from \cref{TestJerome} that
$$\beta _t  = \sum_{i=1}^{2} \gamma_i v_i \lambda _i^t>0, \quad \forall t\in\N.$$
As a conclusion, any solution starting in the set ${\cal C}_1$ switches not more than once. We can conclude that, for this example, any solution to the considered system switches not more than twice as some can be initialized in ${\cal C}_2$. The fact that the origin of the considered system is GES follows by Theorem \ref{thm:stability} as $A_1$ and $A_2$ are Schur. 

\begin{remark} For this example, it turns out that a piecewise quadratic Lyapunov function, given by (\ref{eqLyapPWQ}), exists.
Indeed, the LMIs conditions (\ref{eqLMIPWQ}) given in Appendix~\ref{appendix:PWQ} are feasible with
$
P_1 =  \left[ \begin{array}{cc}  
2.5662  &  2.6008\\
    2.6008  &  2.6427  \end{array} \right]$ , $ 
P_2 = \left[ \begin{array}{cc}  
0.1744 &   0.1570\\
    0.1570  &  0.1482  \end{array} \right] 
$, $
Y_1 = 0.0091$, $Y_2 = 9.2206\times10^{-6}
$, $
Y_3 =  \left[ \begin{array}{cc}  
0.0017   & 0.0232 \\
    0.0232  &  0.0038   \end{array} \right]$, $
Y_4 =  \left[ \begin{array}{cc}  
0.0043  &  0.0058 \\
    0.0058 &   0.0297     \end{array} \right] 
$.
\end{remark}

\subsection{Insulin infusion problem}\label{subsect:applications-insulin}

We consider system (\ref{eqCLS}) with $n=4$, $m=2$ and matrices~$A$,~$B$ and~$K$ given by~(\ref{eq:matrices-insulin-A-B}) and~(\ref{eq:matrices-insulin-K}). Without the results presented in this paper, stability analysis of this $4^{{th}}$-order example is an open question as far as we know. 
Let $A_1=A+BK$ and $A_2=A$ as in Section \ref{subsect:ex-cls-system}. The eigenvalues of $A_2$ are $\lambda_1 =  0.9592$, $\lambda_2 = 0.9512$, $\lambda_3 = 0.9277$ and $\lambda_4 =0.9200$. We next show that the conditions of Theorem~\ref{unboundtn} hold with {$i_1 =2$}.
To this end, we use Theorem~\ref{TestJerome}  along with conditions given in~(\ref{corro4}). 

We start by considering (\ref{eq:St}) with {$i_1 = 1$}, $n = 4$ and $t_1=1$. We have to determine the values of $t_2$ and $t_3$ that satisfy the conditions of Theorem~\ref{unboundtn}, namely for every $t_2,t_{3}\in\N^\star$, we either have  $\Sigma_{i_1;1,t_2,t_{3},0}={\{0\}}$ or there exists an invertible matrix $\mathcal{N}\in\R^{4\times 4}$, whose lines are taken from~\eqref{eq:St} such that items~(i) and~(ii) of Theorem~\ref{unboundtn} are satisfied. By employing linear programming techniques, we have identified 58 pairs of values for $t_2$ and $t_3$, as displayed in \cref{demo-table}.
	\begin{table}[H]
	\begin{center}
		\begin{tabular}{|@{\,}c@{\,}||c|c|c|c|c|c|c|c|c|c|c|c|c|c|c|}
			\hline
			$t_2$ & 1   &  3   &  4  &   5    & 6    & 7     &8    & 9   & 10   & 11   & 12  &  13    
			\tabularnewline
			\hline
			$t_3$  &  1     &1 &    1  &   1  &   1   &  1   & 1 &    1     &1 &    1 &    1 & 1      
			\tabularnewline
			\hline
			\hline
			$t_2$ & 14  &  15   & 16   & 17   & 18  &  19  &20   & 21  &  22    &23    &24 &    1     
			\tabularnewline
			\hline
			$t_3$ &  1  &   1   &  1    & 1   &  1 &    1 &     1   &  1   &  1  &   1  &   1  &   2  
			\tabularnewline
			\hline
			\hline
			$t_2$ &   3 &    4  &   5 &    6  &   7  & 8   &  9   & 10   & 11   & 12   & 13   & 14 
			\tabularnewline
			\hline
			$t_3$ &    2   &  2  &   2   &  2 &    2  &   2     & 2     &2  &   2   &  2  &   2 &    2   
			\tabularnewline
			\hline
			\hline
			$t_2$ &  15   & 16&    17   &  1   &  2  &   3  &   4  &   5     & 6     &7  &   8  &   9  
			\tabularnewline
			\hline
			$t_3$ &   2 &    2  &   2   &  3  &   3 &    3    & 3   &  3  &   3  &    3    & 3  &   3
			\tabularnewline
			\hline
			\hline
			$t_2$ & 10    & 1   &  2  &   3  &   4  &   1 &     2   &  3     &1    & 2 & & 
			\tabularnewline
			\hline
			$t_3$ &   3   &  4 &    4 &    4  &   4    & 5    &  5  &   5   &  6    & 6& &
			\tabularnewline
			\hline
		\end{tabular}
	\end{center}
	\caption{\label{demo-table}Values of $t_2$ and $t_3$ in section \ref{subsect:m=2}.}
\end{table}

For all these 58 pairs, we checked that Theorem~\ref{unboundtn} holds, that is, we have found for each of them an invertible matrix $\mathcal{N}\in\R^{4\times 4}$ whose lines are taken from~\eqref{eq:St} such that items~\ref{it:unboundtn:1} and~\ref{it:unboundtn:2} of Theorem \ref{unboundtn} are satisfied. 
For item~\ref{it:unboundtn:2}, we use Theorem~\ref{TestJerome} and conditions (\ref{corro4}). Detailed computations for two of these pairs are provided in Appendix \ref{appendix:details-checking-conditions-insulin}.   The conclusion is that any solution to the considered system exhibits no more than 4 switches. As the eigenvalues of $A_1$ and $A_2$ lie inside the unit circle, the origin is GES by Theorem~\ref{thm:stability}. 

\section{Conclusion}\label{sect:conclusions}

This paper focuses on the stability analysis of  CLS, whose solutions all exhibit a finite number of switches.
This property is very useful when investigating stability, as we showed, and somehow move the problem to analytically establish that all solutions indeed switch a finite number of times. We have first presented general, sufficient conditions in terms of sets intersections.
To illustrate how these conditions can be exploited, we have then concentrated on the case where two cones partitions the state space, for which we have developed numerically tractable sufficient conditions.
Interestingly, these contributions required the development of novel results on the non-negativity of solutions to discrete-time dynamical systems for given initial conditions.
The theoretical developments have been applied successfully to the optimization-based control of a model of insulin infusion, for which all the existing stability results we are aware of failed. 

It would be interesting in future work to develop similar tools to prove that solutions exhibit a finite number of switches and exploit this property for stability analysis for other classes of hybrid dynamical systems.  
{Another potential direction would be proving the existence of a finite number of switches using the LCS approach. This would first require expressing the CLS as a LCS, followed by the derivation of stability conditions for discrete-time LCS. This path is very challenging. Indeed, stability analysis for discrete-time LCS remains underdeveloped, with limited attention in the literature. }

\appendix
\subsection{Sampled-data modeling for optimal insulin infusion}\label{appendix:sampled} 
Let $y$ be the BGL response, $u$ the insulin input and $f=F \delta$ an impulse of food at time $t=0$. In \cite{Good2019} impulses are used as a mathematical abstraction of insulin applied, and meals consumed and the optimal insulin infusion model is given by linear transfer functions as, for any $s\in\mathbb{C}$, $y(s) = T_F(s) f(s) - T_I(s) u(s)$ where
{\begin{equation}\label{eq:T-I}
T_I(s) := \dfrac{K_I}{({a}_1s+1)({a}_2s+1)({a}_3s+1)}, 
\end{equation}}
with $K_I = 600 \text{ (mmol/L)/(U/min)}$, ${a}_1 = 60 \text{ min}$,  ${a}_2 = 100 \text{ min}$, ${a}_3 = 120 \text{ min}$,  and 
{$$ 
T_F(s) := \dfrac{K_F}{({b}_1s+1)({b}_2s+1)({b}_3s+1)}
$$}
with $K_F = 50 \text{ (mmol/L)/(g/min)}$,  ${b}_1 = 70 \text{ min}$,  ${b}_2 = 110 \text{ min}$ and ${b}_3 = 125 \text{ min}$. To enhance disturbance rejection,  we add an integral action, however instead of implementing a pure integral action, we opt for a first-order filter, for any $s\in\mathbb{C}$, $\tilde{y}(s) = y(s)/(s+\tau)$ with $\tau = 0.015$.
To design the state feedback control, we use a state space model with four state variables, three of them correspond to the transfer function $T_I$ in (\ref{eq:T-I}). It is given by
$\dot{x} = A_c x + B_c u$,  $y = C_c x$ with 
$$
A_c := \left[ \begin{array}{cccc}
   -0.0350  & -0.0249 &  -0.0114   &      0 \\
    \phantom{-}0.0156  &       0     &    0   &      0\\
         0  &  \phantom{-}0.0078      &   0   &      0\\
         0   &      0 &  -3.4133 &  -0.0150 \end{array} \right],
$$
$$
B_c :=\left[ \begin{array}{c}
     2 \\
     0 \\
     0\\
     0\end{array} \right] , \quad 
C_c :=\left[ \begin{array}{cccc}
         0   &      0  & -3.4133    &     0\end{array} \right].
$$
We consider the quadratic cost 
$$
J_c(x_0,u)  := \int_0^{\infty} \Big( x^{\top}(\tau )Q_cx(\tau) + u^{\top}(\tau)R_c u(\tau) \Big) d\tau,
$$
with $Q_c:=\diag(1,1,1,70)$ and $R_c=1$. The design is performed using a sampled-data version, with a sampling period of $T=5$ min, that is (\ref{eqdisc}) with $A =  e^{A_cT}$,  $B = e^{A_cT} B_c$, and cost $J$  as in (\ref{eq:cost-J}) with 
$Q = \int_0^T (e^{A_c\nu})^{\top} Q_c e^{A_c\nu} d\nu$, $S =  \int_0^T  (e^{A_c\nu})^{\top}   Q_c e^{A_c\nu} B_c   d\nu$, $R=\int_0^T B_c^{\top} (e^{A_c\nu})^{\top}  Q_c e^{A_c\nu} B_c  d\nu +  R_c$.  
This leads to the matrices in (\ref{eq:matrices-insulin-A-B}), and $R = 17.8719$,\linebreak
$
Q = 10^4\left[ \begin{array}{rrrr}
    0.0004 &   0.0000 &   0.0009 &  -0.0001 \\
    0.0000  &  0.0012 &   0.0474&   -0.0036 \\
    0.0009  &  0.0474  &  3.2140  & -0.2772 \\
   -0.0001&   -0.0036 &  -0.2772  &  0.0325
\end{array}\right]$ and\linebreak
$S = \left[ \begin{array}{r}
    8.4359 \\
    0.1558 \\
   18.8537\\
   -1.3612
\end{array}\right].
$

\subsection{LMI conditions for conewise quadratic Lyapunov function (\ref{eqLyapPWQ}) with the same partition as (\ref{sysPWQ})}\label{appendix:PWQ}
The next result reduces the conservatism of the conditions in \cite{Feng2002} by taking into account the switches from one cone to another, and obviously applies to general CLS with $m=2$.
\begin{proposition}\label{propPWQ}
	The origin of  {(\ref{sysPWQ})-(\ref{eqcvx12})} is GES if there exist symmetric matrices $P_1$, $P_2$, and symmetric matrices $Y_i$, $i\in\{1,\ldots,4\}$ with non-negative entries, satisfying 
	\begin{equation} \label{eqLMIPWQ}
		\begin{array}{l}
			P_1-\left[\begin{array}{c} K \\KA_1\end{array} \right]^{\top} Y_1\left[\begin{array}{c} K \\KA_1\end{array} \right] >0\\
			A_1^{\top} P_1A_1-P_1+\left[\begin{array}{c} K \\KA_1\end{array} \right]^{\top} Y_1 \left[\begin{array}{c} K \\KA_1\end{array} \right]  < 0\\
			P_2-\left[\begin{array}{c} K \\ KA_2 \end{array} \right] ^{\top}Y_2\left[\begin{array}{c} K \\ KA_2 \end{array} \right] >0 \\
			A_2^{\top} P_2A_2-P_2+\left[\begin{array}{c} K \\ KA_2 \end{array} \right] ^{\top}  Y_2 \left[\begin{array}{c} K \\ KA_2 \end{array} \right] < 0\\
			A_2^{\top} P_1A_2-P_2+ \left[\begin{array}{c} -K \\ KA_2 \end{array} \right] ^{\top}  Y_3 \left[\begin{array}{c} -K \\ KA_2 \end{array} \right] < 0\\
			A_1^{\top} P_2A_1-P_1
				+\left[\begin{array}{c} K \\-KA_1\end{array} \right]^{\top} Y_4 \left[\begin{array}{c} K \\-KA_1\end{array} \right] < 0.
		\end{array}\hskip -1cm
	\end{equation}
\end{proposition}
\begin{proof}
	The proof uses  (\ref{eqLyapPWQ}) as a Lyapunov function candidate.
    The next conditions ensure that $x\mapsto x^{\top}P_i x$, $i\in\{1,2\}$, strictly decreases along solutions to (\ref{sysPWQ}), for $x\neq 0$,
    \begin{align*}
	   \left\{\begin{array}{r@{\>}c@{\>}l}
        x^{\top} P_1x &>& \0,  \\
		x^{\top}(A_1^{\top} P_1A_1-P_1)x&<&\0,
	   \end{array}\right. &&& \quad 	x \in {\cal C}_1 \mbox{ and } A_1 x\in{\cal C}_1,\\
	   \left\{\begin{array}{r@{\>}c@{\>}l} 
		x^{\top} P_2x &>& \0,   \\ 
		x^{\top} (A_2^{\top} P_2A_2-P_2)x&<&\0,  
	   \end{array}\right. &&& \quad 	x \in {\cal C}_2 \mbox{ and } A_2 x\in{\cal C}_2,\\
	   x^{\top} (A_2^{\top} P_1A_2-P_2)x<\0,\ \, &&& \quad x\in {\cal C}_2 \mbox{ and }   A_2x  \in {\cal C}_1,\\
	   x^{\top}(A_1^{\top} P_2A_1-P_1)x <\0,\ \, &&& \quad x\in{\cal C}_1	\mbox{ and } A_1 x\in{\cal C}_2. 
    \end{align*}
	The desired result is obtained by using the S-procedure, which {leads} to the LMIs conditions in (\ref{eqLMIPWQ}).
\end{proof}

\subsection{Discrete-time version of the results in \cite{angeli-homogeneous99}}\label{appendix:attractivity-ges}
The aim of this appendix is to prove that global attractivity of the origin is equivalent to the origin to be GES for homogeneous discrete-time systems of degree $1$, which includes system (\ref{eqCLS}). We consider for this purpose the system
\begin{equation}\label{eq:sys}
	x_{t+1} = f(x_t),
\end{equation}
where $x_t\in\R^{n}$ is the state at time $t\in\N$, $n\in\N^\star$, and $f:\R^{n}\to\R^{n}$ satisfies the next assumption.

\begin{assumption}\label{ass:f}
	Vector field $f$ in \eqref{eq:sys} verifies the properties:
	\begin{enumerate}
		\item[(i)] $f$ is continuous on $\R^{n}$. 
		\item[(ii)] $f$ is homogeneous of degree $1$, i.e., for any $x\in\R^{n}$ and any $\lambda\in\R_+$, $f(\lambda x)=\lambda f(x)$. 
	\end{enumerate}
\end{assumption}

We denote the solution to \eqref{eq:sys} initialized at $x\in\R^{n}$ at time $t\in\N$ as $\phi(t,x)$.
The goal is to prove the next result, which is stated in \cite{angeli-homogeneous99} for continuous-time systems.
Note that the result below is invoked in \cite{sun-tac08,alonso-rocha-tac10} with no proofs.
\begin{theorem}\label{thm:attractivity-ges}
	Consider system \eqref{eq:sys} and suppose that Assumption~\ref{ass:f} holds. The following statements are equivalent.
	\begin{enumerate}[label={(\roman*)}]
		\item\label{it:thm:attractivity-ges:1} $x=0$ is globally attractive, i.e., any solution $\phi$ asymptotically converges to $0$ as time grows.
		\item\label{it:thm:attractivity-ges:2} $x=0$ is GES. 
	\end{enumerate}
\end{theorem}

Before proving Theorem \ref{thm:attractivity-ges}, we need to state the next lemmas. 
Define for $x\in\R^n$ and $\Omega$ an open set of $\R^n$,
$$\tau(x,\Omega)=\inf\left\{ t\in\N\ \,:\,\ \phi(t,x)\in\Omega \right\},$$
with the convention that $\inf\emptyset=\infty$.
The next lemma is a discrete-time special case of \cite[Corollary~III.3]{Sontag-tac96(new-iss)}.
\begin{lemma}\label{lem:sup-finite}
	Consider \eqref{eq:sys} and suppose the following holds.
	\begin{enumerate}
		\item[(i)] Item (i) of Assumption \ref{ass:f} holds.
		\item[(ii)] There {exist} ${\mathcal{D}}\subset\R^{n}$ compact, $\Omega\subseteq\R^{n}$ open, $\mathcal{S}$ compact with $\mathcal{S}\subset \Omega$ and
			\begin{equation}\label{eq:cor-condition}
				\forall x\in{\mathcal{D}},\ \exists t\in\N\text{ s.t. }\ \phi(t,x)\in \mathcal{S}.
			\end{equation}
	\end{enumerate}
	Then there {exists} $t_0\in\N$ such that $\tau(x,\Omega)\le t_0$ for every $x\in{\mathcal{D}}$.
\end{lemma}
\begin{proof}
	Assume by contradiction that there exist $(x_m)_{m\in\N}\in{\mathcal{D}}^{\N}$ such that $\tau(x_{{m}},\Omega)\ge m$. 
	Since ${\mathcal{D}}$ is compact, up to an extraction, we can assume that $(x_m)_m$ converges to some $x^*\in{\mathcal{D}}$.
	But according to \eqref{eq:cor-condition}, there exist $t^*\in\N$ such that $\phi(t^*,x^*)\in\mathcal{S}$.
	In particular, there exists $\varepsilon>0$ such that $\mathbb{B}(x^*,\varepsilon)\subset\Omega$.
	But for every $m>t^*$, we have $\phi(t^*,x_m)\not\in\Omega$, hence, $|\phi(t^*,x_m)-\phi(t^*,x^*)|\ge\varepsilon$.
	This leads to a contradiction with $x\mapsto\phi(t^*,x)$ continuous on $\R^n$.
\end{proof}
The second and final lemma is stated without proof, as it directly follows from the {continuity assumption made on~$f$}.
\begin{lemma}\label{lem:bounded}
	Consider system \eqref{eq:sys} and suppose that item (i) of Assumption \ref{ass:f} holds.
	Let $\Omega$ be a bounded set of $\R^n$ and let $t_0\in\N$.
	There exist $c>0$ depending on $\Omega$ and $t_0$ such that $|\phi(t,x)|\le c$, for every $x\in\Omega$ and every $t\in\{0,\dots,t_0\}$,
\end{lemma}
We are ready to prove Theorem~\ref{thm:attractivity-ges}.
\begin{proof}[Proof of Theorem \ref{thm:attractivity-ges}]
	The implication (ii)$\Rightarrow$(i) directly follows from the corresponding statements. We therefore focus on proving (i)$\Rightarrow$(ii). 
	By homogeneity, it is enough to prove the implication for $|x|=1$. 
	\cref{lem:sup-finite} together with item (i) of Theorem \ref{thm:attractivity-ges} ensure the existence of $t_0\in\N$ and of $t_{x,0}\in\{0,\dots,t_0\}$ such that $\phi(t_{x,0},x)\in\mathbb{B}(0,1/2)$.
	Now, for every $t\in\N$, $\phi(t_{x,0}+t,x)=\phi(t,\phi(t_{x,0},x))=|\phi(t_{x,0},x)|\phi(t,y)$ for $y\in\R^n$, such that $|y|=1$ and $|\phi(t_{x,0},x)|y=\phi(t_{x,0},x)$.
	Hence, there exists $t_{y,0}\in\{0,\dots,t_0\}$ such that $|\phi(t_{y,0},y)|\le1/2$, i.e., $|\phi(t_{x,0}+t_{y,0},x)|\le1/4$.
	Thus, using the homogeneity relation, and repeating the above argument, we conclude to the existence of $(t_{x,\ell})_{\ell\in\N}\in\N^{\N}$ such that $t_{x,\ell}\le t_{x,\ell+1}\le t_{x,\ell}+t_0$ and $\phi(t_{x,\ell},x)\in\mathbb{B}(0,2^{-\ell})$ for every $\ell\in\N$.
	Finally, using the homogeneity, together with Lemma~\ref{lem:bounded}, we conclude to the existence of a constant $c=c\left(t_0\right)$ such that for every $\ell\in\N$ and every $t\in\{t_{x,\ell},\dots,t_{x,\ell+1}\}$, $|\phi(t,x)|\le c2^{-\ell}$.
	That is to say that $|\phi(t,x)|\le c2^{-\nu_x(t)}$, where $\nu_x(t)=\card\left\{ \ell\in\N\ \,:\,\ t_{x,\ell}\le t \right\}$.
	Observe that $\nu_x(t)\ge \lfloor t/t_0\rfloor\ge (t/t_0)-1$.
	We conclude the proof by setting $c_1=2c$ and $c_2=\frac{\ln 2}{t_0}$.
\end{proof}

\subsection{Details for the example of Section~\ref{subsect:applications-insulin}}\label{appendix:details-checking-conditions-insulin}
To illustrate how the conditions of Theorem \ref{unboundtn} are checked, we treat two pairs of $(t_2,t_3)$ among the 58 pairs: $(t_2,t_3)=(1,1)$ and $(t_2,t_3)=(3,2)$.

For $t_2=1$ and $t_3=1$, we have from~\eqref{eq:St}\footnote{Recall \cref{unboundtn} is used with $t_1=1$, and recall that in \eqref{eq:St}, $N_{i_1;t_1,\dots,t_{p},t}$ is defined for $0\le t<t_{p+1}$. Hence, in \eqref{eq:St:exp}, the variable $t$ is $0$ in any cases.}
\begin{align}
	\nonumber
N_{1;t} & =  K ,\\
\label{eq:St:exp}
N_{1;t_1,t} & = -KA_{1} ,\\
\nonumber
N_{1;t_1,t_2,t} & =  KA_{2}A_{1} ,\\
\nonumber
N_{1;t_1,t_2,t_{3},0} & =  -KA_{1}A_{2}A_{1} . 	\nonumber
\end{align}
This leads to
{\small $$
{\cal N} =
\left[ \begin{array}{rrrr} 
   -0.4936 &  -6.9988& -104.7360 &   1.6626 \\
0.0139 &  -3.5448 & -93.6330   & 2.0398 \\
-0.0248    &9.9744  &246.4813 &  -4.9790 \\
-0.0140  &  1.8974 &  69.7493 &  -1.8635
\end{array} \right] .
$$}
To check that the intersection (\ref{eq:interS}) is {reduced to $\{0\}$} for every $t \in \N^\star$, we have to check that
$\mathcal{N}x\ge0$ implies  $-KA_{2}^t\mathcal{M}x\ge0,$ for all $t \in \N^\star$ or equivalently $-KA_{2}^t\mathcal{M}=\beta_t^{\top} \mathcal{N}$, for all $t \in \N^\star$.
We obtain
{\small $$
 	\beta_0 
= \left[ \begin{array}{cccc}   0.0100   & 3.3166  &  1.0539   & 3.7582  \end{array} \right] ^{\top},
$$}
{\small $$
{\cal L}\  = \left[ \begin{array}{rrrr} 
0.0000  &   0.0000   & -0.0000  &   0.0100 \\
-77.6570 &   -0.0000  &   0.0000  &   3.3166 \\
-22.1708  &  -0.0129  &   0.0000  &   1.0539 \\
300.8926  &  22.1708 &  -77.6570  &   3.7582
\end{array} \right],  
$$}
{\small $$
\gamma _1 v_1 = 10^5\left[ \begin{array}{c}
 0.0090    \\ 
2.2552    \\ 
0.7100   \\ 
0.8637    
\end{array} \right]  , \quad \gamma _2 v_2 = 10^5\left[ \begin{array}{c}
    -0.0148   \\   
   -3.6872    \\  
  -1.1608    \\ 
   -1.4041     
\end{array} \right], 
$$
$$
\gamma _3 v_3 = 10^5\left[ \begin{array}{c}
 0.0141    \\ 
   3.4783    \\ 
   1.0951   \\ 
    1.3027   
\end{array} \right]  , \quad \gamma _4 v_4 = 10^5\left[ \begin{array}{c}
    -0.0083\\ 
   -2.0462\\ 
  -0.6442\\ 
   -0.7622
\end{array} \right]. 
$$}
As $\beta _0  = \sum_{i=1}^{4} \gamma_i v_i \lambda _i ^0 >0$, $\lambda_1 > \lambda _2>\lambda_3 > \lambda _4>0$, and conditions~(\ref{corro4}) hold, we conclude using Theorem~\ref{TestJerome} that  
$$
\beta _t  = \sum_{i=1}^{4} \gamma_i v_i \lambda _i^t>0, \quad \forall t \in \N.
$$ 

Now, we consider the case $t_2=3$ and $t_3=2$. The number of possible matrices having 4 lines that can be generated from~\eqref{eq:St} is 35. We have to check if there is an invertible matrix  ${\cal N}$ that allows to conclude that the intersection (\ref{eq:interS}) is {reduced to $\{0\}$} for every $t \in \N^\star$. We found,
{\small $$
{\cal N}=
\left[ \begin{array}{rrrr}      -0.4936   &-6.9988 &-104.7360  &  1.6626\\
	0.0139  & -3.5448  &-93.6330  &  2.0398\\
	0.0248 &  -9.9744 &-246.4813  &  4.9790\\
	-0.0210  & 11.0313  &262.8501 &  -4.9946
\end{array} \right].  
$$}
As previously, we compute
{\small $$
	\beta_0 
	= \left[ \begin{array}{cccc}     0.0017 &   0.0607  &  0.3860  &  2.7102 \end{array}\right]^{\top},
$$ }
{\small     $$
  {\cal L}\  = \left[ \begin{array}{cccc} 
  1.1997  & -0.0049  &-0.00076  &  0.00166\\
  210.63    &  -1.1799   &  -0.42343   &  0.0607 \\
  -239.25  &     1.9792    &   1.0282    &  0.386 \\
  32044   &   -286.92   &   -98.232    &   2.7102
    \end{array} \right],  
$$}
{\small   $$
  \gamma _1 v_1 = 10^4\left[ \begin{array}{c}
  0.0122\\  
1.0926\\  
 0.8372\\  
 1.8402  
  \end{array} \right]  , \quad \gamma _2 v_2 = 10^4\left[ \begin{array}{c}
 -0.0201   \\  
 -1.7998   \\  
 -1.3368   \\  
 -2.9338    
  \end{array} \right],
  $$
  $$
  \gamma _3 v_3 = 10^4\left[ \begin{array}{c}
    0.0190  \\   
    1.7366  \\   
    1.1679  \\  
   2.5604  
  \end{array} \right]  , \quad \gamma _4 v_4 = 10^4\left[ \begin{array}{c}
  -0.0112  \\    
-1.0295  \\   
-0.6683    \\  
-1.4666    
  \end{array} \right].
  $$}
	We have $\beta _0  = \sum_{i=1}^{4} \gamma_i v_i \lambda _i ^0 >0$, $\lambda_1 > \lambda _2>\lambda_3 > \lambda _4>0$, and we verified that conditions (\ref{corro4}) are satisfied. We conclude that  
  $$
  \beta _t  = \sum_{i=1}^{4} \gamma_i v_i \lambda _i^t>0, \quad \forall t\in\N.
  $$
	The same procedure has been repeated for all the 58 possible values of $t_2$ and $t_3$ meaning that  we were able to exhibit a matrix ${\cal N}$ such that, for every  $t \in\N^\star $, the intersection  (\ref{eq:interS}) is {reduced to $\{0\}$} in all these cases. 
    
\bibliographystyle{IEEEtran}
\bibliography{Helly_biblio}

\begin{IEEEbiography}[{\includegraphics[width=1in,height=1.25in,clip,keepaspectratio]{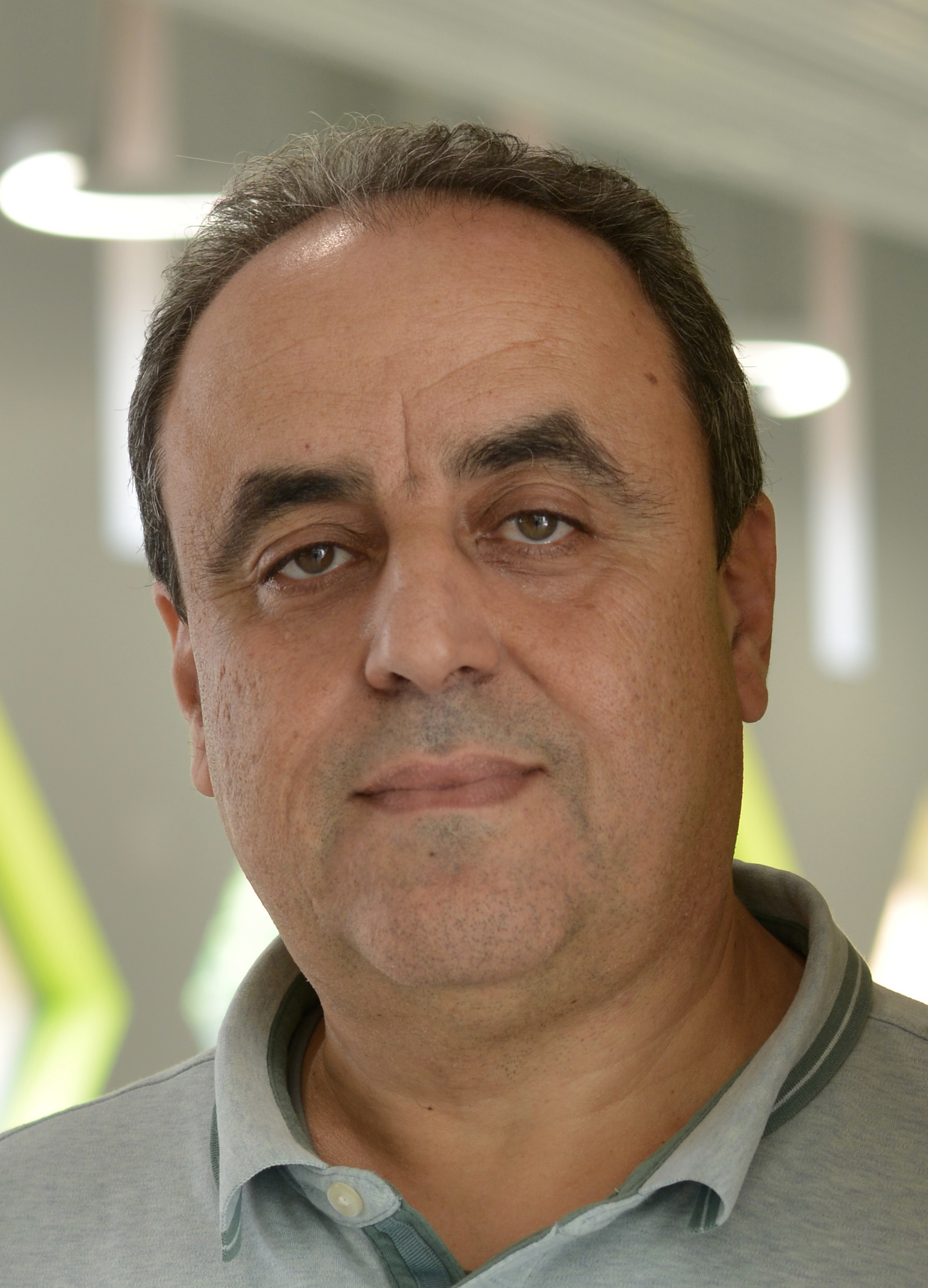}}] {Jamal Daafouz} is a Full Professor at the University of Lorraine (France) and researcher at CRAN-CNRS. He is a senior member of the Institut Universitaire de France (IUF), having been appointed as a junior member in 2010. He received his Ph.D. in Automatic Control from INSA Toulouse and LAAS-CNRS in 1997 and his "Habilitation à Diriger des Recherches" from the University of Lorraine in 2005. His research focuses on the analysis, observation, and control of uncertain, switched and networked systems, with a particular interest for convex based optimisation methods. He has served as an associate editor for several journals, including Automatica, IEEE Transactions on Automatic Control, the European Journal of Control, and Nonlinear Analysis: Hybrid Systems. Currently, he is a senior editor for the IEEE Control Systems Letters.
\end{IEEEbiography}
\begin{IEEEbiography}
	[{\includegraphics[width=1in,height=1.25in,clip,keepaspectratio]{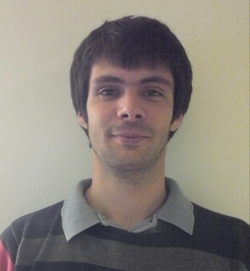}}]{J\'er\^ome Loh\'eac} was born in 1985.
	In 2012, he defended his Ph.D.
	in Applied Mathematics at University of Lorraine, France, under the direction of M.~Tucsnak and J.-F.~Scheid.
	Then, he was post-doc at BCAM, Bilbao (Spain), under the supervision of E.~Zuazua.
	He gets the position of junior researcher at CNRS, France in 2014 and moved to IRCCyN, Nantes.
	Since 2017, he is researcher at CRAN, Nancy.
	His research interests are control and optimal control theory in finite and infinite dimension.
\end{IEEEbiography}
\begin{IEEEbiography}
	[{\includegraphics[width=1in,height=1.25in,clip,keepaspectratio]{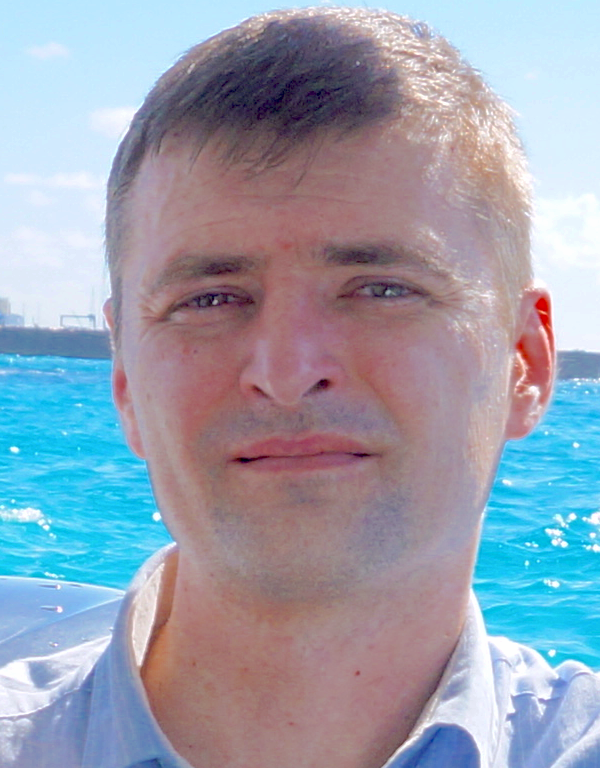}}]{Irinel-Constantin Mor\u{a}rescu} is currently Full Professor at Universit\'e de Lorraine and researcher at CRAN, CNRS in Nancy, France. He received the Ph.D. degree in Mathematics and in Technology of Information and Systems from University of Bucharest and University of Technology of Compiegne, respectively in 2006. He received the ``Habilitation \`a Diriger des Recherches'' from the Universit\'e de Lorraine in 2016. His works concern stability and control of time-delay systems, stability and tracking for different classes of hybrid systems, consensus and synchronization problems. He is in the editorial board of Nonlinear Analysis: Hybrid Systems, IEEE Control Systems Letters, and member of the IFAC Technical Committee on Networked Systems.
\end{IEEEbiography}
\begin{IEEEbiography}
	[{\includegraphics[width=1in,height=1.25in,clip,keepaspectratio]{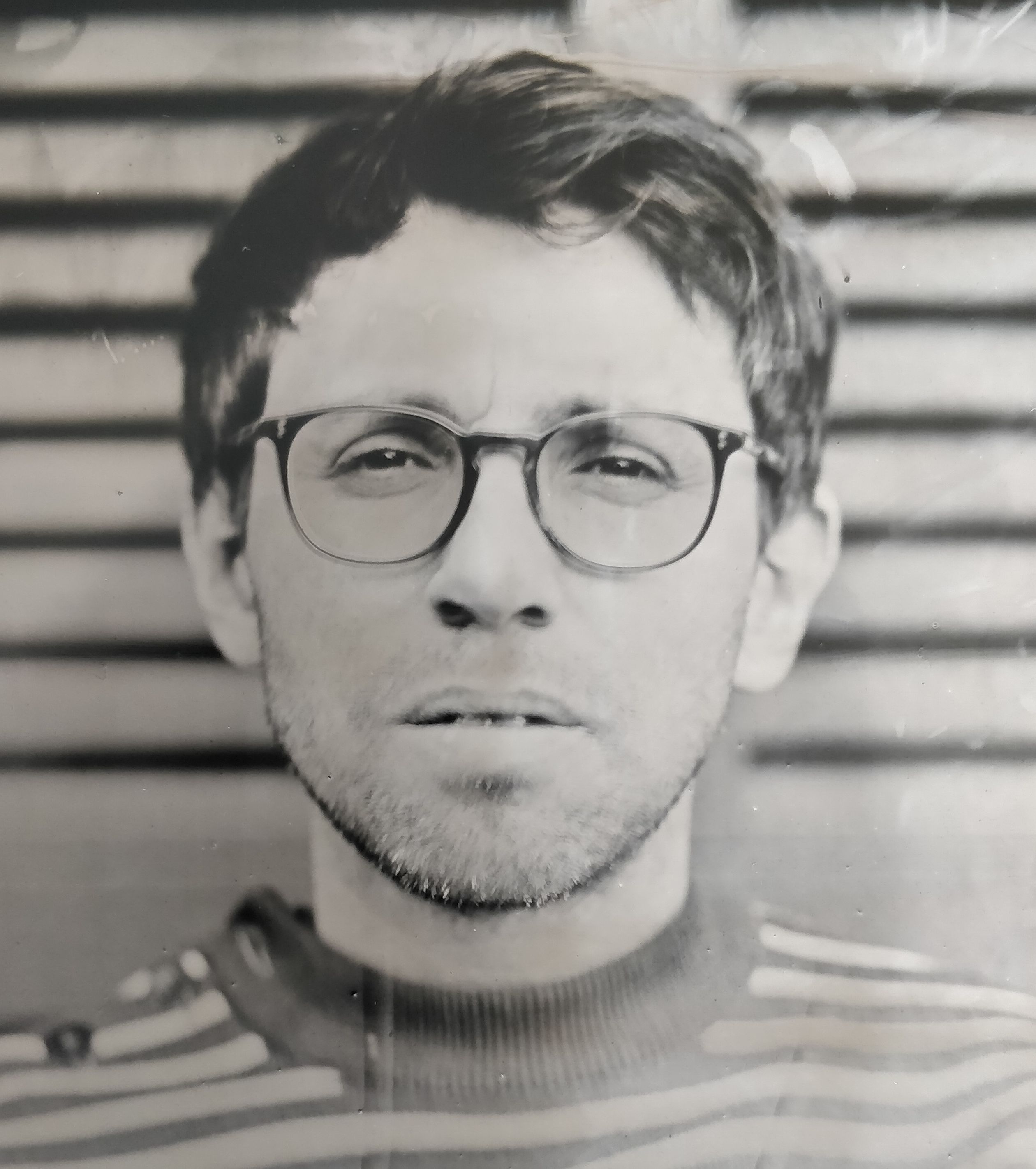}}]{Romain Postoyan} received the ``Ing\'enieur'' degree in Electrical and Control Engineering from ENSEEIHT (France) in 2005. He obtained the M.Sc. by Research in Control Theory \& Application from Coventry University (United Kingdom) in 2006 and the Ph.D. in Control Theory from Universit\'e Paris-Sud (France) in 2009. In 2010, he was a research assistant at the University of Melbourne. Since 2011, he is a CNRS researcher at CRAN, Nancy (France). He received the `Habilitation \`a Diriger des Recherches (HDR)'' from Universit\'e de Lorraine in 2019. He serves/served as an associate editor for the journals: IEEE Transactions on Automatic Control, Automatica, IEEE Control Systems Letters and IMA Journal of Math. Control and Information.
\end{IEEEbiography}
\vskip 0pt plus -1fill

\end{document}